\documentclass[12pt,draftcls,onecolumn]{IEEEtran}
\usepackage{latexsym, epsfig, amsthm, amssymb, amsmath,pxfonts}
\usepackage{showidx}
\usepackage{tikz}

\usepackage{latexsym, epsfig, amssymb, amsmath, tikz, subfigure}
\usepackage{verbatim}
\usepgflibrary{arrows}
\usepackage{verbdef}
\usepackage{color}
\usepackage{graphicx}
\usepackage[european]{circuitikz}
\usepackage{amsmath}
\usepackage{arydshln}
\verbdef{\vtext}{http://arxiv.org/abs/1303.2760}

\tikzset{>=latex}

\tikzset{%
block/.style    = {draw, thick, rectangle, minimum height = 3em,
    minimum width = 3em},
  block1/.style    = {draw, thick, rectangle, minimum height = 1.7em,
    minimum width = 1.7em,fill=gray!70},
      block2/.style    = {draw, thick, rectangle, minimum height = 1.7em,
    minimum width = 1.7em,fill=gray!30},
  sum/.style      = {draw, circle, node distance = 1.8cm}, 
}

\title{Optimal distributed control for platooning via sparse coprime factorizations$^\star$ \thanks{$\star$ Submitted for publication in May 2015. Provisionally accepted for publication as a Full Paper in the IEEE Transactions on Automatic Control.}}

\author{\c{S}erban Sab\u{a}u$^\sharp$,\thanks{$^\sharp$\c{S}erban Sab\u{a}u
is with the Electrical and Computer Engineering Dept.,
Stevens Institute of Technology.
email: ssabau@stevens.edu}
Cristian Oar\u{a}$^\ddag$,\thanks{$^\ddag$Cristian
Oar\u{a} is with the Automatic Control and Systems Engineering Dept., University ``Politehnica'' Bucharest. email: cristian.oara@acse.pub.ro. The work of Cristian Oar\u{a} has been supported by a grant of the Romanian National Authority for Scientific Research, CNCS -  UEFISCDI, project identification number PN-II-ID-PCE-2011-3-0235.}
Sean Warnick$^\dagger$ \thanks{$\dagger$ Sean Warnick is with the Information and Decision Algorithms Labs. and  the Computer Science Department at Brigham Young University. email: sean.warnick@gmail.com.}
and Ali Jadbabaie$^\flat$ \thanks{$^\flat$Ali Jadbabaie is with the Systems Engineering Dept., University of Pennsylvannia. email: jadbabai@seas.upenn.edu.}}

\renewcommand{\tilde}{\widetilde}

\newcommand{\FF}{{{\rm I \kern -0.2em R}}}
\newcommand{\RR}{{{\rm I \kern -0.2em R}}}

\newcommand{\CC}{{{\mbox{\rm \hspace*{0.05ex}
\rule[.18ex]{.18ex}{1.24ex} \kern -.65em C}}}} 

\newcommand{\bea}{\begin{eqnarray}}
\newcommand{\eea}{\end{eqnarray}}

\newtheorem{theorem}{Theorem}[section]
\newtheorem{rem}[theorem]{Remark}
\newtheorem{prop}[theorem]{Proposition}
\newtheorem{lem}[theorem]{Lemma}
\newtheorem{coro}[theorem]{Corollary}
\newtheorem{defn}[theorem]{Definition}
\newtheorem{assumption}[theorem]{Assumption}

\newcommand{\ba}{\left[ \begin{array}}
\newcommand{\baa}{\begin{array}}
\newcommand{\ea}{\end{array} \right]}
\newcommand{\eaa}{\end{array}}

\newcommand{\be}{\begin{equation}}
\newcommand{\ee}{\end{equation}}
\newcommand{\bb}{\begin{equation}\label}

\newcommand{\la}{s}

\def\math#1{\ifmmode{#1} \else {$#1$}\fi}

\newcommand{\sg}{\ifmmode \Sigma \else $\Sigma$ \fi}

\date{}

\begin{document}
\maketitle

\begin{abstract} We introduce a novel distributed control architecture for heterogeneous platoons of linear time--invariant autonomous vehicles.  Our approach is based on a generalization of the concept of {\em leader--follower} controllers for which we provide a Youla--like parameterization, 
while the sparsity constraints are imposed on the controller's left coprime factors, outlying a new concept of structural constraints in distributed control. 
The proposed scheme is amenable to optimal controller design via norm based costs, it guarantees string stability and eliminates the accordion effect from the behavior of the platoon. We also introduce a synchronization mechanism for the exact compensation of the time delays induced by the wireless broadcasting of information.
\end{abstract}

\section{Introduction}

Formation control for platooning of autonomous vehicles has been a longstanding problem in control theory for almost fifty years, going back to the early days of {\em intelligent vehicle highway systems}  \cite{TAC1966}. Since no available control solution was  deemed completely satisfactory, considerable research efforts are still  being spent \cite{Bamieh,Jovanovic,Kan,rick,nasol} motivated by the advent of assisted driving systems and the imminence of driverless vehicles.

The automated control system's objective in the platooning problem is to regulate the inter--vehicle spacing distances (to a pre--specified value) in the presence of disturbances caused by the road and traffic conditions. The problem could be completely solved within the classical  control framework, under the assumption that each vehicle has access, in real time, to an accurate measurement of its relative positions with respect to all its predecessors in the string (centralized control). It became clear from the very beginning that this scenario  is infeasible from several engineering practice standpoints, therefore the control strategies investigated in the literature look only at the situation in which the controller on board each vehicle has access to local measurements only.

The most common premise is that the measurement available to each agent is the instantaneous distance with respect to the vehicle in front of it (measured using onboard sensors), resulting in a control strategy dubbed {\em predecessor follower}. Although (under the standard assumption of linear dynamics for each vehicle) the internal stability of the aggregated platoon can be achieved, this basic architecture was proved to exhibit a severe drawback known as ``string instability'' \cite{Hedrick}. While several 
formal definitions of string instability exist \cite{TAC1996}, they essentially describe the phenomenon of amplification downstream the platoon of the response to a disturbance at a single vehicle. Correspondingly, we will designate as ``string stable'' those feedback configurations for which the $\mathcal{H}_\infty$ norm of the transfer function from the disturbances at any given vehicle to any point in the  aggregated closed--loop of the platoon, does not formally depend on the number of vehicles in the string  \cite{TAC2010}.

If the vehicles dynamics contain a double integrator, then for  predecessor follower schemes of homogeneous platoons  with identical sub--controllers, string instability will  occur irrespective of the chosen linear control law  \cite{Hedrick}, as it is an effect of fundamental limitations of the feedback--loop. This shortcoming cannot be overcome by adding the relative distances with respect to multiple preceding vehicles to the measurements available to each sub--controller  (multiple look--ahead schemes) \cite{Cook, bidirect}, nor can it be overcome by including the successor's relative position (bi--directional control)  \cite{Peppard, Vinnicombe}, without exacerbating the so--called {\em accordion effect} (or settling time) \cite{TAC2010}. The heterogeneous controller tuning  proposed \cite{Brogliato,Davison, Hespanha} offers some benefits for string stability but only at the steep expense of the {\em integral absolute error} specification \cite{TAC2010}. The authors of \cite{TH1,TH2,TH3} proved that (unlike constant inter--vehicle policies) a class of inter--spacing policies dependent of the vehicle's velocity  (dubbed ``time--headways'') can achieve string stability, but only for sufficiently large time--headways which will impair the ``tightness'' of the formation.

A more elaborate, optimal control approach to platooning was also investigated, but the issues pertaining to the increase in size of the platoon persist. In \cite{Bamieh} optimal quadratic regulators for platooning are proposed while showing that for an increasing number of vehicles the resulted LQR problems become ill--posed. It was later proved in \cite{Mitra} that ``local'' measurements based distributed controllers cannot achieve ``coherent'' coordination in large--sized platoons,  results further extended in \cite{Jovanovic} as to achieve superior coherence formation via optimal controllers.

Remarkable performance in terms of both string stability and sensitivity to disturbances can be achieved by the so--called {\em leader--follower} policies \cite{Hedrick}, in which each member of the string has access to the state of the leader's vehicle or an estimate  of the leader's state. However, this approach raises the immediate concern of eventual disruptions in the broadcast of the  leader's state to the follower vehicles. Furthermore, the comprehensive analysis done in \cite{rick} shows that the performance of {\em leader--follower} schemes entailing the transmission the leader's state or its estimate is irremediably altered by the presence of the communications delays induced by the physical limitations of existing wireless systems.

A particularly interesting control architecture \cite{Ploeg2} (named Cooperative Adaptive Cruise Control -- or CACC)  was recently proposed and further adapted as to include an $\mathcal{H}_\infty$ optimality criterion \cite{nasol}. The scheme is based on the elegant results earlier reported in \cite{Naus}, where each vehicle broadcasts its acceleration to its successor in the platoon. However, the performance of the control algorithm proposed in \cite{nasol} is compromised by the presence of (wireless) communications induced delays \cite{nasol2}, since string stability can only be achieved for time--headways  policies,  in accordance with the classical results reported in \cite{TH1,TH2,TH3}. The experimental validation from \cite{nasol2} shows that even for small latencies of the wireless communications systems ({\em e.g.} $20$ milliseconds), relatively large time--headways are needed in order to guarantee string stability.

\subsection{Contributions of This Paper}

In this paper we provide a novel distributed control architecture for heterogeneous platoons of linear time--invariant autonomous vehicles.  We introduce a generalization of the concept of {\em leader--follower} controllers for which we provide a Youla--like parameterization. The structural constraints imposed on the distributed controller can be recast as  sparsity constraints on the Youla parameter, resulting in the tractability of the optimal controller synthesis via $\mathcal{H}_2/\mathcal{H}_\infty$ norm based costs. The distributed implementation allows for the sub--controller on board each vehicle to use only  information from its predecessor in the string.  The proposed architecture is able to compensate the communications induced time delays and can be implemented using existing  high accuracy GPS time base synchronization mechanisms. Such synchronization mechanisms will entail fixed, commensurate and point--wise time delays, thus avoiding the inherent difficulties caused by time--varying or stochastic  or distributed delays. Our approach improves on existing methods in the following essential aspects:
\begin{itemize}
\item[$\bullet$] guarantees string stability in the presence of time delays induced by the wireless communications \cite{nasol2, nasol,rick}; 
\item[$\bullet$] eliminates the {\em accordion effect} from the behavior of the platoon  \cite{TAC2010};  
\item[$\bullet$] achieves string stability, even for constant inter--spacing policies \cite{nasol2, Naus}; 
\item[$\bullet$] allows for optimal controller design via norm based costs, while accommodating heterogeneous  vehicles equipped with heterogeneous controllers \cite{TAC2010, nasol, rick}.
\end{itemize}

Classical methods in distributed/decentralized control formulate the structural constraints on the controller  as sparsity constraints on its transfer function matrix. In turn, our approach formulates certain  sparsity constraints on the controller's left coprime factors \cite{R,CDC,SW1} (that have no meaningful implication on the sparsity of the controller's transfer function matrix), thus outlying a novel concept of structural constraints in the distributed control of multi--agent systems. It is precisely this particular type of constraints on the coprime factors of the controller that induces the distributed implementation of resulted controllers as a network of linear  time--invariant subsystems, such that the sub--controller on board each vehicle uses only  information from its predecessor in the string. This approach to distributed controllers  as linear dynamical networks hinges on the concept of {\em dynamical structure functions}, originally introduced in \cite{Sean08, SW2} and further developed in \cite{SW3,SW4,SW5,SW6,SW7,SW8,SW9,SW10, SW11}.

In addition, we provide a unifying analysis to platooning control, detailing the intrinsic connections of our scheme with the {\em leader--follower} control policies \cite{Hedrick}, with the CACC design  \cite{nasol, nasol2} and  with previous results in distributed/decentralized control such as  {\em quadratic invariant} architectures \cite{Michael's_I3E}. Our analysis concludes that for platooning control  the only  ``local'' measurements needed at each agent in the string are: the inter--spacing distance with respect to its predecessor and the predecessor's control signal, to be used in conjunction with the knowledge of the predecessor's dynamical model. This is an important point since it clarifies previous conjectures \cite[Section~V--B]{nasol},\cite[pp.~5]{TRB}, \cite{nasol2} that additional  information from multiple predecessors (``beyond the direct line of sight'') might lead to superior performance, since they provide a ``preview of disturbances''.

\subsection{Paper Organization}

Section~\ref{PGF} introduces the notation and the instrumental expressions of the doubly coprime factorization within the standard unity feedback control scheme, while in Section~\ref{TPMCP} we provide the precise formulation of the platooning control problem as a disturbances attenuation problem and we also briefly review the {\em predecessor follower} and {\em leader--follower} control policies for platooning.   In Section~\ref{MR} we introduce the concept of {\em leader information} controller, we provide the class of all such controllers associated with a given platoon of vehicle, we discuss the controller's distributed implementation and we point out the intrinsic connections with {\em leader--follower} type policies. Section~\ref{tuneup} is dedicated to the design methods for {\em leader information} controllers, outlying the inherent structural properties of the scheme but also the achievable performance in disturbances attenuation, quantified via  $\mathcal{H}_2/\mathcal{H}_\infty$  norm--based costs. 
Section~\ref{delaydelay} presents a synchronization based mechanism that can completely compensate for the communications induced time delays specific to the physical implementation of {\em leader information} controllers. A comprehensive analysis detailing the underlying connections with previously studied platooning control strategies and with existing distributed/decentralized control architectures including {\em quadratic invariance} is  performed in Section~\ref{soa}. A numerical example displaying the benefits of our novel control scheme is presented 
in Section~\ref{nex}, while Section~\ref{concl} draws the conclusions.

\section{Preliminaries and General Framework} \label{PGF}
\subsection{Basic Notation}
Most of the notation we use  in this paper is quite standard in the systems and control literature.
The Laplace transform complex variable is $s\in \mathbb{C}$ and the Laplace transform of the real signal $u(t)$  will be typically denoted with  $u(s)$ and can be distinguished by the change in the argument. 
When the time argument $\cdot (t)$ or the frequency argument $\cdot (s)$ can be inferred from the context or is irrelevant, it is omitted.

 Table~\ref{NotationTRIL} contains notation for certain  structured  matrices which will be used in the sequel. We also assume the following notation:

\begin{table}
\renewcommand{\arraystretch}{0.8}
\caption[m1]{Notation for structured matrices }
\label{NotationTRIL} \centering
\begin{tabular}{|c|c|}
\hline
$\quad \quad  \mathcal{D}\big\{ d_1,d_2,\dots,d_n \big\} \quad  \quad$ &
$\quad \quad \quad \quad    \ba{ccccc} d_1  & 0 & 0 &   \dots & 0 \\
                   0 &  d_2   & 0 &\dots & 0 \\
                   0 &  0   & d_3 & \dots & 0 \\
                   \vdots & \vdots & \vdots & \ddots & \vdots\\
                   0 &  0   & 0 & \dots & d_n \\  \ea \quad \quad \quad \quad  $ \\

\hline
$\quad \quad  \mathcal{T}\big\{ t_1,t_2,\dots,t_n \big\}  \quad  \quad$ &
$\quad \quad \quad \quad   \ba{rrrrrr} 
                   t_1 &  0   & 0 &\dots &0 & 0 \\
                   t_2 &  t_1   & 0 & \dots & 0& 0 \\
                      t_3 &  t_2   & t_1 & \dots & 0& 0 \\
                   \vdots & \vdots & \vdots & \ddots &\vdots & \vdots\\
                  t_{n-1} &  t_{n-2}   & t_{n-3}  & \dots & t_1 & 0 \\
                 t_n &  t_{n-1}   & t_{n-2}  & \dots & t_2 & t_1 \\  \ea \quad  \quad \quad \quad $ \\
\hline
$\quad \quad  \mathcal{R}\big\{ r_1,r_2,\dots,r_n \big\}  \quad  \quad$ &
$\quad \quad \quad \quad   \ba{ccccc} r_1  & 0 & 0 &   \dots & 0 \\
                   r_2 &  r_2   & 0 &\dots & 0 \\
                   r_3 &  r_3   & r_3 & \dots & 0 \\
                   \vdots & \vdots & \vdots & \dots & \vdots\\
                   r_n &  r_n   & r_n & \dots & r_n \\  \ea \quad  \quad \quad \quad $ \\

\hline
\end{tabular}
\end{table}


\begin{tabular}{l l} \label{nott}
& \\

$x\overset{def}{=} y$ &  $x$ is by definition equal to $y$ \\
 $\mathbb{R}(\la) $ & Set  of all real--rational  transfer functions.  \\
 $\mathbb{R}(\la)^{p \times q} $ & Set of $p \times q$ matrices having all entries in $\mathbb{R} (\la)$ \\
 LTI & Linear and Time Invariant \\
 TFM & Transfer Function Matrix \\ 
$Q_{ij}$ & The $i$--th row, $j$--th column entry of $Q\in\mathbb{R}(\la)^{p \times q}$ \\
$P \star u(t)$ &  The time response with zero initial conditions of an  (LTI) system \\
&  with TFM $P$ and input $u(t)$ \\
$T_{z_iw_j}$ & The $i$--th row, $j$--th column entry of the TFM $T_{zw}\in\mathbb{R}(\la)^{p \times q}$, mapping \\
&input vector $w$ to output vector $z$ \\ 
 \end{tabular}



\subsection{The Standard Unity Feedback Loop}

We focus on the standard {\em unity feedback}  configuration of Figure~\ref{2Block},
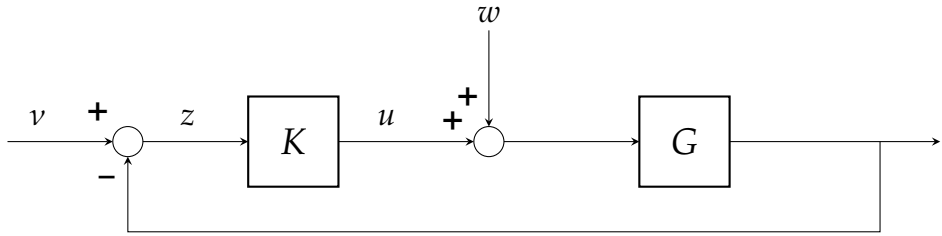
\begin{figure}[h]
\begin{tikzpicture}[scale=0.4]
\draw[xshift=5cm, >=stealth ] [->] (0,0) -- (3.5,0);
\draw[ xshift=5cm ]  (4,0) circle(0.5);
\draw[xshift=5cm] (3,1)   node {\bf{+}} (1,0.8) node {$\nu$};
\draw [xshift=5cm](6,0.8)   node {$z$} ;
\draw[ xshift=5cm,  >=stealth] [->] (4.5,0) -- (8,0);
\draw[ thick, xshift=5cm]  (8,-1.5) rectangle +(3,3);
\draw [xshift=5cm](9.5,0)   node {\large{$K$}} ;
\draw[ xshift=5cm,  >=stealth] [->] (11,0) -- (15.5,0);
\draw[ xshift=5cm ]  (16,0) circle(0.5cm);
\draw [xshift=5cm](12.6,0.8)   node {$u$} ;
\draw [xshift=5cm] (14.8,0.7)   node {\bf{+}};
\draw[  xshift=5cm,  >=stealth] [->] (16,3.7) -- (16,0.5);
\draw [xshift=5cm] (16,3.6)  node[anchor=south] {$w$}  (15.3,1.5)  node {\bf{+}};
\draw[  xshift=5cm,  >=stealth] [->] (16.5,0) -- (21,0);
\draw[ thick, xshift=5cm ]  (21,-1.5) rectangle +(3,3) ;
\draw [xshift=5cm] (22.5,0)   node {\large{$G$}} ;
\draw[ xshift=5cm,  >=stealth] [->] (24,0) -- (31,0);
\draw[ xshift=5cm,  >=stealth] [->] (29,0) -- (29,-3) -- (4,-3) -- (4, -0.5);
\draw [xshift=5cm] (3.3,-1.3)   node {\bf{--}};
\useasboundingbox (0,5);
\end{tikzpicture}
\caption{Standard unity feedback loop of the plant $G$ with the controller $K$}
\label{2Block}
\end{figure}
where $G$ is a multivariable (strictly proper) LTI plant and $K$ is an LTI controller. Here, $w$ and $\nu$  are the input disturbance and sensor noise, respectively and $u$  and $z$ are the controls and measurements vectors, respectively.  
Denote by  
\begin{equation} \label{crazy}
H(G,K) = \ba{cc} T_{zw} & T_{z\nu} \\ T_{uw} & T_{u\nu} \ea \overset{def}{=}
\ba{cc}  -(I+GK)^{-1}G & (I+GK)^{-1}\\  -K(I+ KG)^{-1}G & K(I+GK)^{-1} \ea
\end{equation}
the closed--loop TFM of Figure~\ref{2Block} from the exogenous signals $\displaystyle [w^T \; \; \nu^T\; \: ]^T$ to $\displaystyle [z^T \;\; u^T \;\: ]^T$
We say a certain TFM is {\em stable} if it has all its poles in the  open left complex half--plane, and {\em unimodular}
if it is square, proper, stable and has a stable inverse.
If $H(G,K)$ is  stable we say that $K$ is a {\em stabilizing controller} of $G$, or equivalently  that $K$ {\em stabilizes} $G$. 



\subsection{Coprime and Doubly Coprime Factorization for LTI Systems} \label{2}

Given a square plant $G \in \mathbb{R}(\la)^{n \times n}$, a {\em right coprime factorization} of $G$  is a fractional representation of the form $G = NM^{-1}$ with  both factors $N, M \in \mathbb{R}(\la)^{n \times n}$ being stable and for which there exist $X,Y \in \mathbb{R}(\la)^{n \times n} $ also stable, satisfying ${Y}M + {X}N = I_n$  (\cite[Ch.~4, Corollary~17]{V}), with $I_n$ being the identity matrix. Analogously, a {\em left  coprime factorization} of $G$  is defined by $G = \tilde{M}^{-1}\tilde{N}$,  with $\tilde{N}, \tilde M \in \mathbb{R}(\la)^{n \times n} $ both stable and satisfying $\tilde{M}\tilde{Y} + \tilde{N}\tilde{X} =I_n$, for certain stable TFMs $\tilde{X}, \tilde Y \in \mathbb{R}(\la)^{n \times n}$. 

\begin{defn} \cite[Ch.4, Remark pp. 79]{V} 
A collection of eight  stable TFMs $\big(M, N$, $\tilde M, \tilde N$, $X, Y$, $\tilde X, \tilde Y\big)$
is called a  {\em doubly coprime factorization}  of $G$  if   $\tilde M$ and $M$ are invertible, yield the
factorizations
$G=\tilde{M}^{-1}\tilde{N}=NM^{-1}$, and satisfy the following equality  (B\'{e}zout's identity):
\bb{dcrel}
\ba{cc}   -\tilde N &  \tilde M \\ Y &  X \ea
\ba{cc}  - \tilde X & M \\   \tilde Y & N \ea = I_{2n}.
\ee
\end{defn}

\begin{theorem}
\label{Youlaaa}
{\bf (Youla)} \cite[Ch.5, Theorem 1]{V} Let  $\big(M, N$, $\tilde M, \tilde N$, $X, Y$, $\tilde X, \tilde Y\big)$ be a doubly coprime factorization of $G$. Any controller $K_Q$ stabilizing the plant $G$, in the feedback interconnection of Figure~\ref{2Block}, can be written as
\begin{equation}
\label{YoulaEq}
K_Q=Y_Q^{-1}X_Q = \tilde X_Q \tilde Y_Q^{-1},
 \end{equation}
where
$X_Q$, $\tilde{X}_Q$, $Y_Q$ and $\tilde{Y}_Q$ are defined as:
\begin{equation} \label{EqYoula4}
X_Q  \overset{def}{=}  X+Q \tilde{M}, \quad 
\tilde{X}_Q  \overset{def}{=}  \tilde{X}+MQ, \quad 
Y_Q  \overset{def}{=}  Y - Q \tilde{N}, \quad 
\tilde{Y}_Q  \overset{def}{=}  \tilde{Y}-NQ  
\end{equation}
for some stable $Q$ in $\mathbb{R}(\la)^{n\times n}$. It also holds that $K_Q$ from (\ref{YoulaEq}) stabilizes $G$, for any stable $Q$ in $\mathbb{R}(\la)^{n\times n}$.
\end{theorem}

\begin{rem}
\label{rem:232pm27feb2012}
Starting from any doubly coprime factorization (\ref{dcrel}), the following identity  \begin{equation}\label{dcrelQ}
\ba{cc} -\tilde N &  \tilde M \\  Y_Q&  X_Q  \ea
\ba{cc}  -\tilde X_Q & M \\   \tilde Y_Q & N \ea = I_{2n}.
\end{equation}
provides an alternative doubly coprime factorization of $G$, for any stable $Q\in \mathbb{R}(\la)^{n\times n}$.

\end{rem}

\begin{lem} (\cite[(7)/~pp.101]{V}) \label{HGK}
For any stabilizing controller $K_Q$ from (\ref{YoulaEq}), 
the expression of the closed--loop $H(G,K_Q)$  \cite[(32)/~pp.107]{V}
takes the form \cite[(32)/~pp.107]{V}





\begin{equation} \label{folos}
\ba{cc}  T_{zw} & T_{z\nu}\\  T_{uw} & T_{u\nu} \ea = \ba{cc}  - \tilde Y_Q \tilde N  & \;\; \tilde Y_Q \tilde M \\   -\tilde X_Q \tilde N & \; \;\tilde X_Q \tilde M \ea
\end{equation}
where the block--wise partition in the identity (\ref{folos}) is in accordance with  the definitions of (\ref{crazy}).
\end{lem}

\section{The Platoon Control Problem} \label{TPMCP}


We consider a platoon of one leader and $n\in \mathbb{N}$ follower vehicles traveling in a straight line
along a highway, in the same (positive) direction of an axis with origin at the starting point of the leader.
\underline{Henceforth, the ``$0$'' index  will be reserved for the leader.} 
We denote by  $y_0(t)$ the time evolution of the position of the {\em leader vehicle}, which can be regarded as the ``reference'' for the entire platoon.
The dynamical model for the $k$--th  vehicle in the string, ($0 \leq k \leq n$)  is described by its corresponding LTI, continuous--time, finite dimensional transfer function $G_k(\la)$  from its controls $u_k(t)$ to its position $y_k(t)$ on the roadway. While in motion, the $k$--th vehicle is affected  by the disturbance  $w_k(t)$, additive to the control input $u_k(t)$, specifically
\begin{equation} \label{dacia}
\displaystyle y_k(t)=G_k\star\big(w_k(t)+u_k(t)\big).
\end{equation}
For the leader's vehicle  we make  the distinct specification that the control signal $u_0(t)$ is not assumed to be automatically generated (we do not assume the existence of a controller on board the leader's vehicle). Actually, both $u_0$ and $w_0$ act as {\em reference} signals for the entire platoon.


The goal is for every vehicle in the string to follow the leader while maintaining a certain inter--vehicle spacing distance which we denote with $\Delta$. If the inter--vehicle spacing policy is assumed to be constant then $\Delta$ is given as a pre--specified {\em positive} constant. Under the standard assumptions \cite{Hedrick, nasol, rick} that all vehicles start at rest ($\dot{y}_k(0)=$ for $0\leq k \leq n$) and from  the initial desired formation  ($y_k(0)=-k\Delta$ for $0\leq k \leq n$),  the time evolution for the position of each vehicle becomes  \cite[(1)/ pp. 1836]{Hedrick}:
\begin{equation} \label{vehicle}
y_k(t)=G_k\star\big(u_k(t)+w_k(t)\big) - k \Delta, \; \mathrm{for} \quad 0\leq k \leq n.
\end{equation}

We denote with $z_k(t)$ the inter--vehicle spacing {\em errors} defined as 
\begin{equation} \label{spacing}
z_k(t)\overset{def}{=}y_{k-1}(t) -y_k(t) - \Delta, \; \mathrm{for} \quad 1\leq k \leq n.
\end{equation}

The objective of the control mechanism is to attenuate the effect of the disturbances $w_k$, ($0\leq k \leq n$), and of the leader's control signal  $u_0$ at each member of the platoon, such as to maintain the spacing errors (\ref{spacing}) as close to zero as possible.
This ``small errors'' performance  must be attained asymptotically (in steady state) and for a constant speed of the leader. The error signals relate to the performance metrics associated with the platoon (as an aggregated system) when considering safety margins and traffic throughput.

\begin{rem} \label{delta} There is no loss of generality in assuming that $\Delta=0$ in equation (\ref{vehicle}) 
or in considering vehicles with different lengths, since these parameters can be ``absorbed'' as needed in the spacing error signals (\ref{spacing}). These assumptions are standard in the literature  \cite{Hedrick, nasol, rick}, they do not alter the subsequent analysis, and are introduced hereafter for illustrative simplicity.
\end{rem}

In practice an inter--vehicle spacing policy that is proportional with the  vehicle's speed $\dot{y}_k(t)$ (dubbed {\em time headway}) is preferred to the {\em constant} policy (\ref{spacing}). Time headway policies  \cite{TH1,TH2,TH3} have been known to have beneficial effects  on  certain stability measures of the platoon's behavior. For a {\em constant time headway}  $h>0$,  the expression of the spacing errors becomes
\begin{equation} \label{spacinghead}
z_k(t)\overset{def}{=}y_{k-1}(t) -\big(y_k(t)+h \: \dot{y}_k(t)\big) 
\end{equation} 
where $\dot{y}_k(t)$ is the speed of the $k$--th vehicle.\footnote{Note that for $h=0$ in (\ref{spacinghead}) the  time--headway becomes the constant vehicle inter--spacing policy (\ref{spacing}).}
Under the aforementioned ``zero'' error initial conditions \cite[Section~II]{Hedrick} 
we can write the vehicle inter--spacing  errors as:


\begin{equation}\label{errors}
z_{k+1}=G_{k}\star (u_{k}+w_{k})- HG_{k+1}\star (u_{k+1}+w_{k+1}) , \quad 0 \leq k \leq (n -1),
\end{equation}
where
\begin{equation} \label{Ha}
H(s)\overset{def}{=}hs+1, \quad h > 0.
\end{equation}


Next, we will use the following standard notation for the aggregated signals of the platoon
\begin{equation}
z\overset{def}{=} \ba{cccc} z_1& z_2 & \dots & z_n\ea^T, \quad
w\overset{def}{=} \ba{cccc} w_1& w_2 & \dots & w_n\ea^T, \quad
u\overset{def}{=} \ba{cccc} u_1& u_2 & \dots & u_n\ea^T.
\end{equation}
Define  $T\in \mathbb{R}(\la)^{n \times n}$ as
\begin{equation} \label{guma}
T\overset{def}{=}\ba{ccccc}   H& O& O & \dots & O \\
                      -1 & H & O& \dots & O \\
                     O & -1 & H & \dots &O \\
                     \vdots & \vdots & \vdots &  &\vdots \\
                     O & O & O & \dots & H    \ea
\end{equation}
 \noindent while noting  that its inverse is 
  \begin{equation} \label{gumainv}
T^{-1}\overset{}{=}\mathcal{T}\big\{ H^{-1}, H^{-2}, \dots, H^{-n}\big\}= \ba{ccccc}   H^{-1}& O& O & \dots & O \\
                      H^{-2} & H^{-1} & O& \dots & O \\
                     H^{-3} & H^{-2} & H^{-1} & \dots &O \\
                     \vdots & \vdots & \vdots &\ddots  &\vdots \\
                     H^{-n} & H^{-n+1} & H^{-n+2} & \dots & H^{-1}    \ea.
\end{equation}

\subsection{Platoon Motion Control as a Disturbance Attenuation Problem}




Rewriting (\ref{errors}) for all $0 \leq k \leq (n -1)$ in a matrix form, we obtain

\begin{equation} \label{4BlockEquiv}
\ba{c} z_1 \\  z_2 \\ z_3 \\ \vdots \\ z_n \ea = \ba{c}  1 \\ O \\ O \\ \vdots \\ O \ea G_0 \star (u_0+w_0)
- \ba{ccccc}   HG_1& O& O & \dots & O \\
                      -G_1 & HG_2 & O& \dots & O \\
                     O & -G_2 & HG_3 & \dots &O \\
                     \vdots & \vdots & \vdots &  &\vdots \\
                     O & O & O & \dots & HG_n    \ea \star
                    \ba{c}  (u_1+ w_1) \\ (u_2 + w_2) \\ (u_3+w_3) \\ \vdots \\ (u_n + w_n) \ea.
                    \end{equation}

\begin{defn} \label{Plant}
In view of (\ref{4BlockEquiv}), we will denote with $G\overset{def}{=}T\: \mathcal{D}\{G_1,G_2, \dots, G_n\}$ the aggregated TFM of the platoon, from the controls vector $u$ to the error signals vector $z$. Henceforward, we will refer to $G$  as \underline{{\em the platoon's plant}}.
\end{defn}
With this notation equation (\ref{4BlockEquiv}) can be expressed as

\begin{equation} \label{4BlockEquivbis}
z=V_1 G_0 \star(u_0+w_0) - G \star(u+ w),
\end{equation}
where $V_1\overset{def}{=} \ba{cccc} 1& 0 & \dots & 0\ea^T$ is the first column vector of the Euclidian basis in $ \mathbb{R}^{n}$.

In our platooning framework the measurements of  the platoon's plant are the errors signals $z$, representing  
 the input signals of the controller $K_Q\in \mathbb{R}(\la)^{n\times n}$,  therefore the equation for the controls vector reads
\begin{equation} \label{Kz}
 u(t)=K_Q\star z(t).
 \end{equation}


To bridge the gap between our platooning control problem and the generic unity feedback scheme from Figure~\ref{2Block},  we simply plug (\ref{Kz}) into (\ref{4BlockEquivbis}) in order to obtain the  closed--loop $H(G,K_Q)$ of the platoon (as an aggregated system). 









\begin{prop} \label{cheia1}
Given a doubly coprime factorization (\ref{dcrel}) of the platoon's plant $G$ and  a controller $K_Q$ (\ref{YoulaEq})  then

\begin{subequations}
\begin{equation} \label{tempo1}
z=T_{zw_0}\star(u_0+w_0)+ T_{zw}\star w,
\end{equation}
\begin{equation} \label{AS2}
u= T_{uw_0} \star(u_0+w_0) +  T_{uw}\star w,
\end{equation}
\end{subequations}
where $T_{zw_0}\overset{def}{=}\big( I +GK_Q \big)^{-1}V_1 G_0$ and $T_{uw_0}\overset{def}{=}K_Q (I+GK_Q)^{-1}V_1 G_0$ are the TFMs  
from the leader's controls and disturbances  $(u_0+w_0)$  to the interspacing errors $z$ and control signals $u$, respectively,  
while $T_{zw}$ and $T_{uw}$ are as defined in (\ref{crazy}), for $K=K_Q$. In particular, it holds that\footnote{ For clarity of the exposition, the analysis done in this paper employs a slightly different interpretation of the controls signal $u$ than the standard one from \cite{V}. Specifically, in this paper $u$ is the output of the controller without the additive disturbance $w$, such that the input signal of the plant in Figure~\ref{2Block} is $(u+w)$. Therefore the closed--loop TFM $T_{uw}$ has a different expression than the one in \cite[(7)/~pp.101]{V}). The difference is not conceptual but  merely conventional and is needed here for additional simplicity.}
\begin{equation} \label{tempo2}
\ba{cc} T_{zw_0} & T_{zw} \\ T_{uw_0} & T_{uw} \ea= 
\ba{cc} \tilde Y_Q \tilde MV_1 G_0 & -\tilde Y_Q \tilde N \\  
\tilde X_Q \tilde M V_1G_0 & -\tilde X_Q \tilde N \ea.
\end{equation}
\end{prop}
\begin{proof} Plug (\ref{Kz}) into (\ref{4BlockEquivbis}) in order to get  (\ref{tempo1}).
(The expression in (\ref{tempo1}) can also  be retrieved from  \cite[(13)/ pp. 1839]{Hedrick} for the case of identical vehicles.) Next, note that because of (\ref{crazy}) it holds that $T_{zw_0}=T_{z\nu}V_1 G_0$ and substitute accordingly the expression from (\ref{folos}) of Lemma~\ref{HGK} into (\ref{tempo1}) in order to obtain $T_{zw_0}$ in (\ref{tempo2}). By plugging (\ref{Kz}) into (\ref{4BlockEquivbis}) we get that $\displaystyle u=K_Q\big(V_1 G_0 \star (u_0+w_0) - G \star w-G \star u\big)$ which yields (\ref{AS2}). Note that also because of (\ref{crazy}) it holds that $T_{uw_0}=T_{u\nu}V_1 G_0$  and substitute accordingly the expression from (\ref{folos})  of Lemma~\ref{HGK} into (\ref{AS2}) in order to get $T_{uw_0}$ in (\ref{tempo2}). The remaining expressions in (\ref{tempo2}) follow directly from Lemma~\ref{HGK}.
\end{proof}

\begin{rem} \label{namol} Clearly, from (\ref{tempo2}) it appears that the stability of $T_{zw_0}$ or $T_{uw_0}$ cannot be guaranteed by an internally stabilizing controller $K_Q$ for any leader dynamics $G_0$. However, this issue can be solved under lenient assumptions, as explained later in the sequel.
\end{rem}

\subsection{Predecessor Follower Control} \label{PFC}
Proposition~\ref{cheia1} provides the Youla parameterization (convex in the parameter $Q\in \mathbb{R}(\la)^{n \times n}$) of {\em all} closed--loop maps, achievable with stabilizing controllers.  One of the problems specific to the platooning setup is that  the corresponding Youla parameterization yields {\em centralized} controllers $K_Q\in \mathbb{R}(\la)^{n \times n}$ whose TFMs have no particular sparsity pattern whatsoever. In view of equation (\ref{Kz}), this  means that  in order to generate the control signal $u_k$ for any fixed $k$--th vehicle in the platoon ($1\leq k \leq n$), all other measurements $z_j$, with $1\leq j \leq n$, must be available on board the $k$--th vehicle.  Even with today's communications technology this scenario is simply not feasible from multiple engineering standpoints. That is why in the control literature has been extensively studied the more practical scenario in which the controller $K_Q$ from (\ref{Kz}) is constrained to be {\em diagonal}. This  translates into a scheme in which each one of the vehicles in the platoon only needs access to  the spacing error with respect to the vehicle in front of it (measurable with on board ranging sensors). The scheme has been dubbed {\em predecessor following} control and is depicted in Figure~\ref{figprecedede}. The {\em predecessor follower}  scheme  has certain  fundamental drawbacks such as the fact that any diagonal LTI controller $K_Q$ leads to the undesired phenomenon of {\em string instability} \cite{TAC2010}. For an extensive analysis on the subject we refer to \cite{TAC2010} and the references within.

\begin{figure}[h]
\includegraphics[scale=.95]{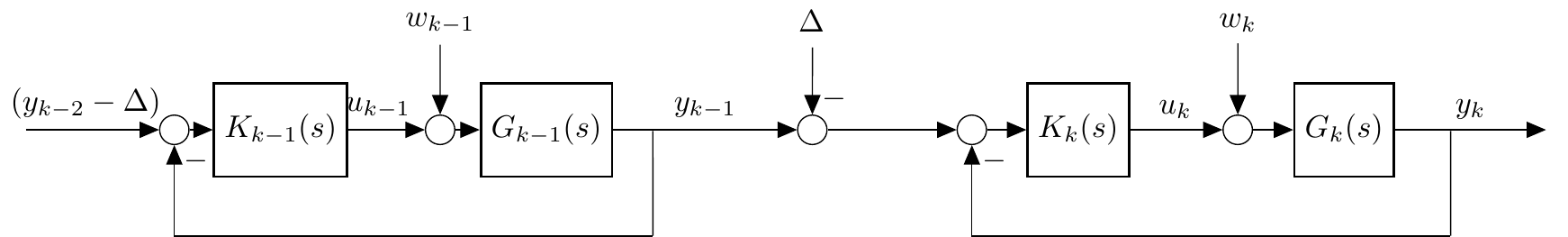}
\caption{Predecessor Follower Control Scheme with Constant Interspacing Policy $\Delta$}
\label{figprecedede}
\end{figure}

\subsection{The Leader Information Control Scheme } \label{LeaderExtension}


In \cite[(11)/ pp. 1838]{Hedrick}  
the case of  platoons with identical vehicles is studied and particular attention is paid to control   laws of the form $u_k=K_k\:(y_0-y_k)$, for $1\leq k \leq n$, where $(y_0-y_k)$ is the relative distance from the $k$--th vehicle to the leader. The intuition behind this control scheme is the fact that the leader's vehicle trajectory $y_0(t)$ is basically the reference for the entire platoon, hence all vehicles in the platoon should ``mimic'' the leader's behavior in order to maintain zero spacing errors.  For the constant inter--vehicle spacing policy (\ref{spacing}) it holds that $\displaystyle (y_0-y_k)=z_1+z_2+\dots +z_k +k\Delta$, therefore under the standard assumptions of Remark~\ref{delta}, writing such control policies in a compact  form yields
\begin{equation} \label{Kell}
\ba{c} u_1 \\  u_2 \\ u_3 \\ \vdots \\ u_n \ea= \ba{ccccc}
                   K_{1} & O & O &\dots &  O \\
                   O&   K_{2}& O & \dots & O \\
                        O &  O &   K_{3} & \dots & O\\
                        \vdots & \vdots &\vdots & \vdots & \vdots \\
                       O &  O &  O&   \dots  & K_n \\ \ea\ba{c} z_1 \\ (z_1+ z_2) \\ (z_1+z_2 + z_3) \\ \vdots \\ (z_1+z_2+z_3+\dots +z_{n-1} +z_n) \ea
\end{equation}
We rewrite equation (\ref{Kell}) such that the input vector is the vector of measurements $z$, in accordance with our Definition~\ref{Plant} of the platoon's plant, obtaining 
\begin{equation} \label{Kell2}
\ba{c} u_1 \\  u_2 \\ u_3 \\ \vdots \\ u_n \ea= \ba{ccccc}
                   K_{1} & O & O & \dots & O \\
                      K_{2}&   K_{2}& O& \dots &O \\
                        K_{3} &  K_{3} &   K_{3} & \dots & O\\
                       \vdots &   \vdots &   \vdots &    \vdots  & \vdots \\
                       K_n & K_n &   K_n & \dots & K_n \ea\ba{c} z_1 \\  z_2 \\ z_3 \\ \vdots \\ z_n \ea.
 \end{equation}

The stabilizing controllers featuring the particular  structure in (\ref{Kell2}) were dubbed {\em leader--follower} controllers or {\em leader information} controllers. An excellent analysis of such control policies can be found in  \cite[(11)/ pp. 1838]{Hedrick}  for the situation where  all vehicles are considered identical,  all controllers are also taken to be identical and a constant inter--spacing policy is implemented. The TFM of the type (\ref{Kell2}) controllers can be retrieved from \cite[(12)/ pp. 1838]{Hedrick} by taking  $K_p=0$ (control {\em without} predecessor information).

The key feature of {\em leader information} control policies is the fact that they can achieve {\em string stability} along with excellent sensitivity to disturbances \cite{Hedrick}. 
In exchange for this, the practical implementation drawbacks  stem from the fact that each one of the $n$ vehicles in the platoon must have  at all times access to a highly accurate measurement of its instantaneous relative position with  respect to the leader, namely $(y_0(t)- y_k(t))$.\footnote{For a platoon comprising of three hundred vehicles traveling at 60 MPH (100 km/h) while maintaining the lawful interspacing distance, the measurement $(y_0(t)- y_k(t))$  for the last vehicles in the platoon is of the order of ten miles (sixteen kilometers). This  renders  very large errors unavoidable when measuring $(y_0(t)- y_k(t))$ also due to the fact that (along the same line of the highway) different vehicles have slightly different trajectories  and therefore they traverse slightly different distances. These errors have major detrimental effects on the control performance.} This aspect is further complicated by the fact that the leader must continuously broadcast its instantaneous coordinates to each vehicle in the platoon and the physical limitations of the (wireless) communications entail delays at the receivers's end.  The presence of communications delays severely deteriorates the control performance \cite{rick}.

\section{Main Result} \label{MR}



We introduce in the next definition a variation of the control law in (\ref{Kell2}) -- called also {\em leader information} --  which 
inherits the performance features characteristic to these controllers for homogeneous strings of vehicles \cite{Hedrick}.


\begin{defn} \label{leader} A controller $K_Q\in \mathbb{R}^{n \times n}(s)$  is said to be a {\em leader information controller}, if $K_Q$ stabilizes the platoon's plant $G$ in the feedback configuration of Figure~\ref{2Block} and the TFM $T_{zw_0}=(I_n +GK_Q )^{-1}$ from the disturbances at the leader to the errors is {\em diagonal}.
\end{defn}

\begin{rem} \label{legatio}
It turns out that imposing sparsity constraints on the closed--loop TFM $(I_n +GK_Q )^{-1}$ (from the disturbances to the  leader $w_0(t)$ to the errors vector $z(t)$) arises as a natural performance condition in multi--agent platooning  systems,  as we argue in detail in  Section~\ref{tuneup}. This is due to the fact that the sparsities of  these closed--loop TFMs are intimately related to the manner in which the disturbances propagate through the string formation.
\end{rem}


The vehicle's linearized  dynamics are commonly modeled in the literature as a second order system including damping \cite{Peppard, Bamieh}, or as a double integrator with first order actuator dynamics \cite{Hedrick,Brogliato}. In this work we do not need to be directly concerned with the  transfer function of the vehicle's dynamical model, however,  we will henceforth operate under the following assumption that allows to model the distinct masses and the distinct actuating  time constants corresponding to the different types of vehicles in the platoon ({\em e.g.} heavy vehicles versus automobiles).

\begin{assumption} \label{A4}
The dynamical model $G_k$ for each vehicle $k$, $0\leq k \leq n$, equals a given {\em strictly proper} transfer function $G_\wp(\la) \in \mathbb{R}(\la)$ weighted by a {\em unimodular} factor $\Phi_k\in \mathbb{R}(\la)$, specifically $G_k\overset{def}{=}\Phi_ k G_\wp.$
We will henceforth denote the following $n \times n$ diagonal unimodular TFM with $\Phi\overset{def}{=}\mathcal{D} \{\Phi_1, \dots, \Phi_n \}$.  The expression of the platoon's plant therefore becomes $G=T\Phi G_\wp$.
\end{assumption}

In particular, for a {\em point--mass} model comprising of the double integrator with a first order actuator ($\tau_k > 0$), 
\begin{equation}
G_k(s)=\dfrac{s+\sigma_k}{m_k\: s^2(\tau_k s+1)}, \quad \mathrm{for} \quad 1 \leq k \leq n,
\end{equation}
\noindent $G_\wp$ would be the double integrator $1/s^2$ and $\Phi_k$ would be equal to $\dfrac{s+\sigma_k}{m_k(\tau_ks+1)}$ with $\sigma_k>0$, where $m_k$ and $\tau_k$ are the mass and actuator time constant respectively, specific to the $k$-th vehicle.
\begin{rem} 
We remark here that since zero at $-\sigma_k$ in the expression of $\Phi_k$ is stable, it actually does not introduce important 
restrictions \cite{Seron}, and can be cancelled in closed--loop.\end{rem}
\begin{rem} \label{Pade}
We can also allow for the transfer function   $G_\wp$ from Assumption~\ref{A4} to include a  conveniently designed Pade  rational approximation of $e^{-\tau s}$, taken to be the same for all $n$ vehicles in the platoon and for the leader. This assumption is made as to take into account an actuation time delay $e^{-\tau s}$ (of the Electrohydraulic Braking and Throttle actuation system -- see  for example \cite{Naus}), with $\tau$ assumed to be the same for all vehicles. The delay $\tau$ is known in practice from the vehicle's technical specifications and can be further verified through model validation methods.  It is also known that the Pade approximation will introduce non--minimal phase zeros (depending on $\tau$) in $G_\wp$  and therefore some loss of performance.
\end{rem}



\subsection{The Youla Parameterization of All Leader Information Controllers}

In this subsection we provide the Youla parameterization of {\em all} leader information controllers associated with a given platoon of vehicles. Our result is formulated in terms of a particular doubly coprime factorization of the platoon's plant, whose factors feature certain sparsity patterns.  As it turns out, parameterizing all leader information controllers  translates into restricting the set of the Youla parameters only to those having a {\em diagonal} TFM.  This feature is remarkably convenient for the optimal  leader information controller synthesis, because it entails a complete ``decoupling'' of the design problem, as later explained in Subsection~\ref{tuneup}. First, we will need the following preparatory result. 
\begin{prop} \label{chelia} The $n\times n$ Transfer Function Matrix $\big(H^{-1}T\big)$ with the constant time--headway $H(s)=1+hs$, (with $h>0$) and $T$ as defined in (\ref{guma}) is a unimodular TFM. 
\end{prop}
\begin{proof} It follows from the fact that $ H^{-1}T$ has all its poles and all its Smith zeros at $-\dfrac{1}{h}$, where $h>0$ as specified in (\ref{Ha}).
\end{proof}

\begin{theorem} \label{leaderinftare}
 Let $\big(M_\wp, N_\wp$, $\tilde M_\wp, \tilde N_\wp$, $X_\wp, Y_\wp$, $\tilde X_\wp, \tilde Y_\wp\big)$ be a doubly coprime factorization of $G_\wp$, where all eight factors  are {\em scalar} rational functions, with $\tilde N_\wp$ and $N_\wp$ strictly proper\footnote{Because $G_\wp$ is assumed strictly proper in Assumption~\ref{A4}}. Then:

{\bf (A)} There exists a doubly coprime factorization (\ref{dcrel}) of $G$, denoted $\big(M, N$, $\tilde M, \tilde N$, $X, Y$, $\tilde X, \tilde Y \big)$, and having the following expression


\begin{subequations} \label{hurray}
\begin{equation} \label{hurraya}
\ba{cc}   -\tilde N &  \tilde M \\ Y &  X \ea \overset{def}{=} \ba{rc}   -\tilde N_\wp T \Phi  &  \tilde M_\wp I_n \\ Y_\wp H^{-1} T \Phi   &  X_\wp H^{-1} I_n \ea,
\end{equation}
\begin{equation} \label{hurrayb}
\ba{cc}  - \tilde X & M \\   \tilde Y & N \ea \overset{def}{=} \ba{cc}  - \Phi^{-1}T^{-1} \tilde X_\wp & \Phi^{-1} T^{-1}HM_\wp  \\   \tilde Y_\wp I_n & HN_\wp I_n \ea;  
\end{equation}
\end{subequations}



{\bf (B)} The Youla parameterization (\ref{YoulaEq}) of {\em all} leader information stabilizing controllers (from Definition~\ref{leader}) is obtained from the doubly coprime factorization (\ref{hurray}) by constraining the Youla parameters $Q\in \mathbb{R}(\la)^{n \times n}$ to be diagonal, specifically $\displaystyle Q\overset{def}{=}\mathcal{D}\big \{Q_{11},Q_{22}, \dots, Q_{nn}\big \}$, with $Q_{kk}\in \mathbb{R}(\la)$ stable, for any $1\leq k \leq n$.   Moreover, any {\em leader information controller} $K_Q$ is given by a left coprime factorization of the form
\begin{equation} \label{ff}
\ba{cc} Y_Q & X_Q \ea \overset{def}{=}\ba{cc} \big( H^{-1}Y_\wp I_n -  \tilde N_\wp Q \big) T \Phi  & \;  \big( H^{-1}X_\wp I_n + \tilde M_\wp Q \big) \\ \ea. 
\end{equation}
The detailed expressions of the factors $Y_Q$ and $X_Q$ are given by
\begin{subequations} \label{exceptzeonal}
\begin{equation} \label{March8th}
Y_Q=\ba{ccccc}
                   (Y_\wp- Q_{11} H \tilde  N_\wp) \Phi_1& O & O & \dots & O \\
                      (-H^{-1}Y_\wp+ Q_{22}\tilde N_\wp)\Phi_1&   (Y_\wp- Q_{22} H\tilde N_\wp)\Phi_2& O& \dots &O \\
                        O &  (-H^{-1}Y_\wp+Q_{33} \tilde  N_\wp)\Phi_2 &   (Y_\wp- Q_{33}H\tilde N_\wp)\Phi_3 & \dots & O\\
                       \vdots &   \vdots &   \vdots &    \vdots  & \vdots \\
                       O & O &   O & \dots & (Y_\wp-Q_{nn} H \tilde N_\wp)\Phi_n \ea,
\end{equation}
\begin{equation}
X_Q=\ba{ccccc}
                   (H^{-1}X_\wp+Q_{11} \tilde M_\wp) & O & O & \dots & O \\
                    O&   (H^{-1}X_\wp+Q_{22} \tilde M_\wp)& O& \dots &O \\
                        O &  O &   (H^{-1}X_\wp+Q_{33} \tilde M_\wp) & \dots & O\\
                       \vdots &   \vdots &   \vdots &    \vdots  & \vdots \\
                       O & O &   O & \dots & (H^{-1}X_\wp+Q_{nn}\tilde M_\wp) \ea.
\end{equation}
\end{subequations}
\end{theorem}
\begin{proof} {\bf (A)} The fact that both $T^{-1}$ from (\ref{gumainv}) and $HT^{-1}$  from Proposition~\ref{chelia} are stable, implies that all eight factors  from (\ref{hurray}) are stable. The rest of the proof  follows by the inspection of (\ref{hurray}) which complies with the definition from (\ref{dcrel}).

{\bf (B)} The TFM of interest $(I_n+GK_Q)^{-1}$ is diagonal if and only if $(I_n+GK_Q)^{-1} (\tilde M_\wp I_n)^{-1} = \tilde Y_Q$ is diagonal (since $\tilde M_\wp$ is not identically zero and by Lemma~\ref{HGK} applied to the factorization in (\ref{hurraya})). 
The latter holds  if and only if and only if  $(\tilde Y_Q - \tilde Y_\wp I_n)$ is diagonal. But from  the expression for $Y_Q$ from (\ref{EqYoula4}) applied to the factorization in (\ref{hurray}) it follows that $\tilde Y_Q=\tilde Y_\wp I_n +  (HN_\wp )Q$ and since neither $N_\wp$, nor $H$ are identically zero, clearly  $\tilde Y_Q$ is diagonal if and only if the Youla parameter $Q$ is a diagonal TFM. The formulas from (\ref{exceptzeonal}) follow by directly employing Theorem~\ref{Youlaaa} to the particular doubly coprime factorization in (\ref{hurray}).
\end{proof}

\subsection{A Distributed Implementation of Leader Information Controllers}

In this subsection we introduce a {\em distributed implementation} for leader information controllers which we will prove to be of great practical interest. Our proposed scheme is based on a natural adaptation of the controller's left coprime factorization from (\ref{exceptzeonal}). First, we note that since the inverse $Y_Q^{-1}$ of the factor from (\ref{March8th}) is lower triangular, it follows that the TFM $K_Q=Y_Q^{-1}X_Q$ of any of the  leader information controllers  parameterized in Theorem~\ref{leaderinftare}  is also lower triangular. This suggests that in order to compute $u_k$ on board the $k$--th vehicle,  we would need access to the interspacing errors $z_j$, with $1\leq j \leq k$, of all vehicles preceding the $k$--th vehicle. As it turns out, our distributed implementation completely circumvents this requirement. The following key result is an immediate consequence of Theorem~\ref{leaderinftare}.


\tikzset{%
block/.style    = {draw, thick, rectangle, minimum height = 1.7em,
    minimum width = 1.7em},
  block1/.style    = {draw, thick, rectangle, minimum height = 1.7em,
    minimum width = 1.7em,fill=gray!70},
      block2/.style    = {draw, thick, rectangle, minimum height = 1.7em,
    minimum width = 1.7em,fill=gray!30},
  sum/.style      = {draw, circle, node distance = 1.4cm}, 
  input/.style    = {coordinate}, 
  output/.style   = {coordinate} 
}
\newcommand{\lp}{$+$}

\begin{figure*}[h]
\hspace{-3mm}
\includegraphics[scale=.4]{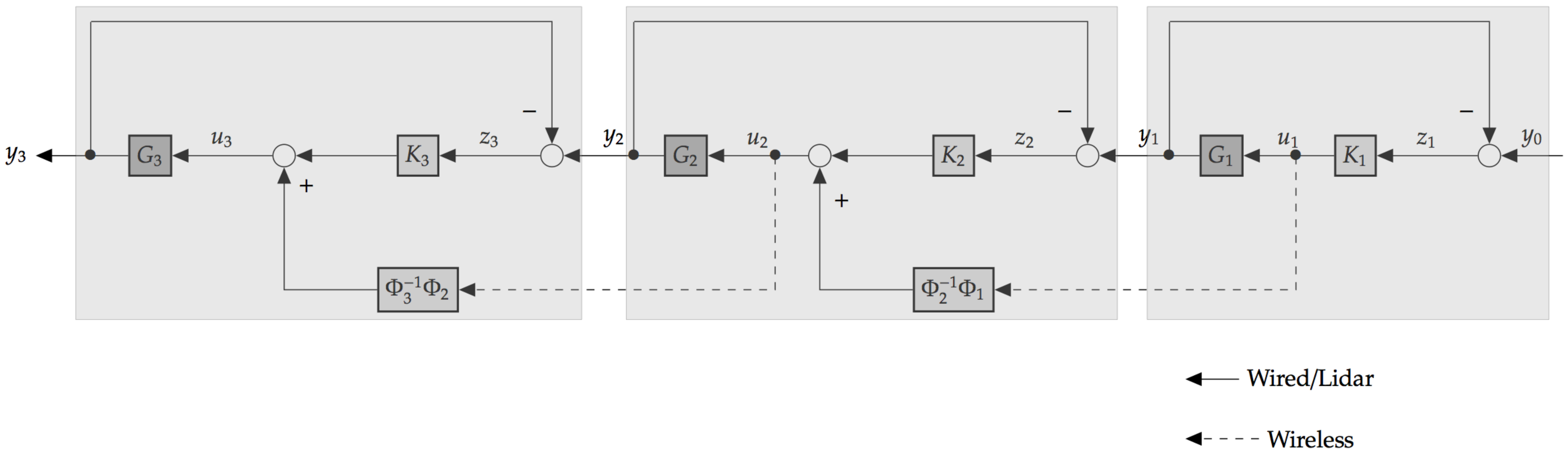}
\caption{Distributed Implementation of the Leader Information Controller}
\label{f2}
\end{figure*}



\begin{figure*}[h]
\centering
{\footnotesize
\begin{tabular}{c}
\cr
 $\ba{c} u_1 \\ u_2 \\ u_3 \\ \vdots \\ u_{n-1} \\ u_n  \ea = \ba{ccccccc}  O & O & O &  \dots & O & O\\ \Phi_2^{-1}\Phi_1 & O & O &  \dots & O & O\\ O & \Phi_3^{-1}\Phi_2 & O &  \dots & O & O \\ O & O & \Phi_4^{-1}\Phi_3  &  \dots & O & O \\ \vdots & \vdots & \vdots & \ddots & \vdots & \vdots \\  O & O & O &  \dots &\Phi_n^{-1}\Phi_{n-1}   & O  \ea \star \ba{c} u_1 \\ u_2 \\ u_3 \\ \vdots \\ u_{n-1} \\ u_n\ea +
 \ba{cccccc}  K_{1} & O & O & \dots  & O & O\\ O & K_{2}  & O & \dots  & O & O \\ O & O & K_{3} & \dots  & O & O \\ O & O & O & \ddots & O & O  \\ \vdots & \vdots & \vdots & \ddots & \vdots & O \\ O & O & O &  \dots & O & K_n \\  \ea \star \ba{c} z_1 \\ z_2 \\ z_3 \\ \vdots \\ z_{n-1} \\ z_n \ea $
 \\\cr
\end{tabular}}
\caption{The Equation for Leader Information Controller from Figure~\ref{f2}}
\label{ecfig}
\end{figure*}

\begin{coro} \label{invent}
Any of the {\em leader information controllers} $K_Q$, $u=K_Qz$, parameterized in Theorem~\ref{leaderinftare} can be rewritten as
\begin{equation} \label{fff}
u=H^{-1} \Phi_{ldiag}\star u+H^{-1} K_{}\star z
 \end{equation}
with
\begin{equation}
\Phi_{ldiag}\overset{def}{=}\ba{ccccccc}  O & O & O &  \dots & O & O\\\Phi_2^{-1}\Phi_1 & O & O &  \dots & O & O\\ O & \Phi_3^{-1}\Phi_2 & O &  \dots & O & O \\ O & O & \Phi_4^{-1}\Phi_3  &  \dots & O & O \\ \vdots & \vdots & \vdots & \ddots & \vdots & \vdots \\  O & O & O &  \dots &\Phi_n^{-1}\Phi_{n-1}   & O  \ea
\end{equation}
\begin{equation} \label{kakaka}
K\overset{def}{=}\mathcal{D} \big \{K_1, K_2, \dots, K_n \big \}, \quad K_k\overset{def}{=}\Phi_k^{-1} \big(Y_\wp-Q_{kk} H \tilde  N_\wp\big)^{-1} (X_\wp+Q_{kk} H \tilde M_\wp), \quad \mathrm{with} \quad K_k\in \mathbb{R}(\la),
\end{equation}

\begin{figure}[h]
\hspace{-3mm}
\includegraphics[scale=.81]{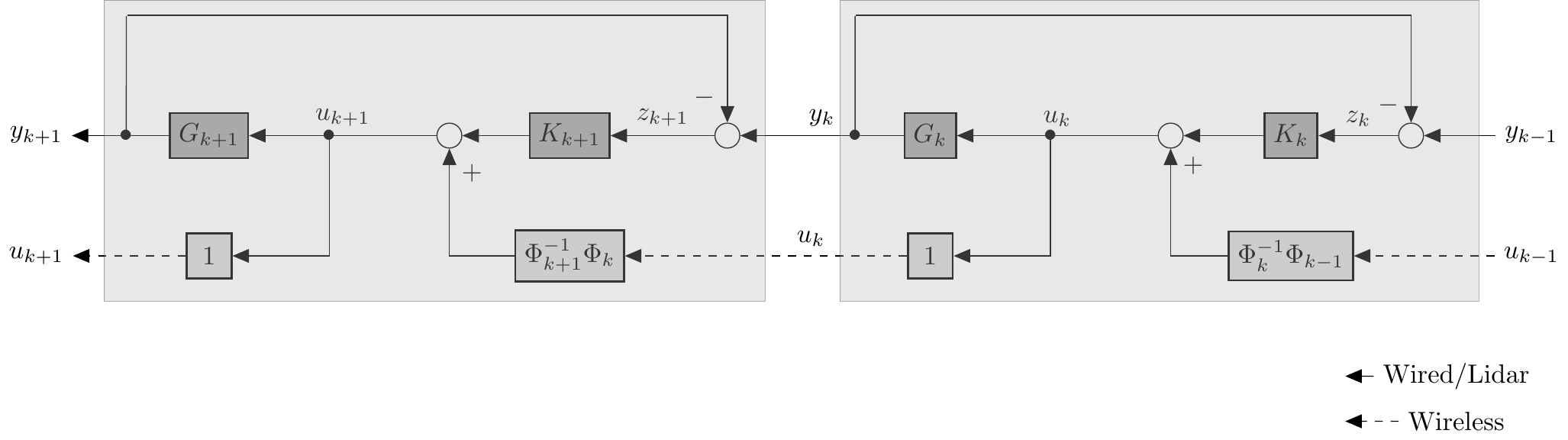}
\caption{Distributed Implementation of Leader Information Control}
\label{f2bis}
\end{figure} 

\noindent for any $1 \leq k \leq n$.

\end{coro}

\begin{proof} It follows directly by multiplying to the left coprime factorization (\ref{ff}) of $K_Q$ with the diagonal TFM
\[
\mathcal{D}\Big\{ \Phi_1^{-1} H^{-1}\big(H^{-1}Y_\wp-Q_{11} \tilde  N_\wp\big)^{-1},  \Phi_2^{-1} H^{-1} \big(H^{-1}Y_\wp-Q_{22} \tilde  N_\wp\big)^{-1}, \dots, \Phi_n^{-1}H^{-1} \big(H^{-1}Y_\wp-Q_{nn} \tilde  N_\wp\big)^{-1}\Big\}
\] \end{proof}



A distributed  implementation of the leader information controller according to  Corollary~\ref{invent}  is presented in Figure~\ref{f2} for the first three vehicles of the platoon, followed by the equation of the leader information controller in Figure~\ref{ecfig}. The scheme for any two consecutive vehicles in the platoon  ($k \geq 2$) is depicted in Figure~\ref{f2bis},\footnote{To make the graphics more readable we  have illustrated in both Figure~\ref{f2} and Figure~\ref{f2bis}, the case in which the constant time--headway policy  has been removed, meaning that we considered $H(s) =1$. However, the implementation for  $H(s) =hs+1$ with $h>0$ should become straightforward from equation (\ref{fff}): simply add a cascaded $H^{-1}$ filter on each  of the $u_k$ and $z_{k+1}$ branches and  add a $H^{}$ filter on each feedback branch from $y_k$ to $z_k$ (in accordance to the definition of the error signals $z_k$ from (\ref{errors})).} where the $\Psi$ blocks  are considered to be equal to $1$. The blocks $G_k$ (with $1\leq k \leq n$) represent the dynamical models of the vehicles and each control signal $u_k$ is fed in the Electrohydraulic Braking and Throttle actuation system on board the $k$--th vehicle. The $z_k$ measurement, represents the distance to the preceding vehicle and it is measured using ranging sensors on board the $k$--th vehicle. The control signal $u_{k-1}$ produced on board the $(k-1)$--th vehicle is broadcasted ({\em e.g.}, using wireless communications) to the $k$--th vehicle.
The blocks $\Psi$ are taken to be $1$ in Figure~\ref{f2bis} because we assume there are no (wireless) communications induced delays.


\subsection{Supplemental Remarks on Definition~\ref{leader}} \label{description}

In order to get some intuitive insight on the content of Definition~\ref{leader}, we must look at the expression of the weighted controls vector $\Phi u$  instead of the controls vector $u$ (with $\Phi$ as defined in Assumption~\ref{A4}). 
We therefore premultiply (\ref{fff}) to the left with $\Phi$ and bring the   $\Phi\star  u$ factors on the left hand side in order to obtain $\displaystyle (H^{-1}T)\: \Phi \star u = (H^{-1}\Phi K) \star z$ or, equivalently, 
\begin{equation} \label{compensa}
\Phi \star u= T^{-1} \Phi K \star z,
\end{equation}

\noindent with $K$ as defined in (\ref{kakaka}) of Corollary~\ref{invent}. To make our point and for this current Subsection~\ref{description} only, let us assume constant inter--spacing policies (\ref{spacing}) (by taking the constant time headway\footnote{See also footnote before equation (\ref{spacinghead}).}  $h=0$ in (\ref{Ha})) and observe that  under this assumption $T$ from (\ref{guma})  satisfies $ T^{-1} =\mathcal{R}\big\{ 1, 1,\dots, 1 \big\}$ such that (\ref{compensa}) becomes
\begin{equation} \label{compensabis}
\Phi \star u=  \mathcal{R}\big\{ 1, 1,\dots, 1 \big\}  \Phi K\:  \star z.
\end{equation}
Since $\Phi$ and $K$ are both diagonal, (\ref{compensabis}) implies that the TFM from the weighted measurements  $\Phi_k K_k \star z_k$ to the weighted controls $\Phi_k \star u_k$ is of the form  (\ref{Kell2}). If, furthermore, in the definition of $K_k$ from (\ref{kakaka}) all the entries of the (diagonal) Youla parameter from Theorem~\ref{leaderinftare} are taken to be identical, that is  $Q_k=Q^*$  for all $1\leq k \leq n$ (with $Q^*$  some fixed, stable transfer function) then

\begin{equation} \label{nicenice}
u=  \mathcal{R} \Big\{ K_1 , K_2, \dots, K_n  \Big \} \star z, \quad K_k\overset{def}{=} \Phi_k^{-1}  \big(Y_\wp-Q^* \tilde  N_\wp \big)^{-1} (X_\wp +Q^* \tilde M_\wp ) \quad \text{for all} \; \; 1\leq k \leq n.
\end{equation}
This shows that our leader information controller from Definition \ref{leader} is indeed a controller of type (\ref{Kell2}), thus  validating  
 its given name. 





\section{Performance of Leader Information Controllers} \label{tuneup}

In this section we deal with the performance characteristics of leader information controllers. The discussion is twofold:
\begin{itemize}
\item[$\bullet$] First, we bring forward a {\em structural} feature of any {\em leader information} controller which determines the  non--propagation of disturbances downstream the platoon.   These results are presented in Subsection~\ref{strruct} next; 

\item[$\bullet$] Second, as the main exploit of the Youla parameterization from Theorem~\ref{invent}, we look at how {\em leader information} controllers perform in achieving  disturbances attenuation (via norm--based costs). This discussion is performed in Subsections~\ref{consid} and ~\ref{Praktika}.
\end{itemize}










\subsection{Structural Properties of Leader Information Controllers} \label{strruct}
As a structural property of {\em any} leader information controller, the resulted closed--loop TFM $T_{zw}$ from Proposition~\ref{cheia1} is lower bidiagonal.  This implies that any disturbance $w_j$ (at the $j$--th vehicle in the platoon) will only impact the $z_j$ and $z_{j+1}$ error signals.  Consequently, any disturbance at the $j$--th vehicle is completely attenuated before even propagating to the $(j+2)$--th vehicle in the string. This phenomenon is in accordance with the analysis done in \cite{Hedrick} on the excellent performance  of  leader--follower control policies  with respect to {\em sensitivity to disturbances} (see also the discussion from Subsection~\ref{description}).

Furthermore, since  according to Definition~\ref{leader} the TFM $T_{zw_0}$ is diagonal, the disturbances $w_0$ at the leader's vehicle influence only the $z_1$ error signal and none of the subsequent errors $z_k$, with $k\geq 2$.\footnote{Similarly, the  leader's control signal $u_0$ impacts only the $z_1$ error signal, and not at all the subsequent errors $z_k$, with $k\geq 2$.  This is relevant to the current discussion, since (as specified in Section~\ref{TPMCP}) $u_0$ is not automatically generated and so it constitutes a reference signal for the entire platoon.}  This feat of leader information controllers practically eliminates the so-called {\em accordion effect} from the behavior of the platoon. In contrast, for any of the predecessor--follower type schemes mentioned in Subsection~\ref{PFC}  (including bi--directional \cite{Peppard, Vinnicombe} or multi look--ahead schemes \cite{Cook, bidirect}), since  $T_{zw}$ is lower triangular, disturbances at the $j$--th vehicle  -- even if attenuated -- affect the inter--spacing errors of all its successors in the platoon, therefore exhibiting the  {\em accordion effect}. The following result provides the exact expressions of the closed--loop TFMs achievable with leader information controllers.

\begin{lem}\label{stringstab} Given a doubly coprime factorization (\ref{hurray})  of the platoon's plant $G$ and $\displaystyle Q\overset{def}{=}\mathcal{D}\big\{Q_{11},Q_{22}, \dots, Q_{nn}\big\}$ a diagonal Youla parameter, 
it holds that:

 {\bf (A)} The closed loop transfer function from the  disturbance $w_0$ 

 \ \ \ \ \ \  {\bf i)} to the interspacing error signals $z_k$ 
 is given by

 \begin{equation} \label{une}
T_{z_kw_0}=\left \{ \begin{array}{cc} (\tilde Y_\wp-HN_\wp Q_{11})\tilde N_\wp\Phi_0, \:  & \quad \mathrm{for} \; k=1, \\ 0, & \quad  \mathrm{for} \; k \geq 2; \end{array} \right.
 \end{equation}


 \ \ \ \ \ \  {\bf ii)} to the control signals $u_k$ 
 is given by

   \begin{equation} \label{deux}
T_{u_kw_0}=   (\tilde X_\wp+HM_\wp Q_{11})\tilde N_\wp \Phi_0\Phi_k^{-1}H^{-k}.
 \end{equation}

 {\bf (B)} The closed loop transfer function from the disturbance $w_j$ 

  \ \ \ \ \ \  {\bf i)} to the error signals $z_k$ 
  is given by
   \begin{equation} \label{troix}
T_{z_kw_j}=\left \{ \begin{array}{cl}  0, & \quad  \mathrm{for} \; k < j,  \\ -(\tilde Y_\wp-HN_\wp Q_{jj})\tilde N_\wp H \Phi_j, \:  & \quad \mathrm{for} \; k=j, \\ (\tilde Y_\wp -HN_\wp Q_{(j+1)(j+1)})\tilde N_\wp \Phi_{j}, \:  & \quad \mathrm{for} \; k=j+1, \\ 0, & \quad  \mathrm{for} \; k > j+1; \end{array} \right.
 \end{equation}

 \ \ \ \ \ \  {\bf ii)} to the control signals $u_k$ 
 is given by






 \begin{equation} \label{quatre}
T_{u_kw_j}=\left \{ \begin{array}{cl}  0, & \quad  \mathrm{for} \;  k<j,  \\ - (\tilde X_\wp +HM_\wp Q_{jj})\tilde N_\wp,  \:  & \quad \mathrm{for} \; k=j, \\ - M_\wp (Q_{jj}-Q_{(j+1)(j+1)})\tilde N_\wp \Phi_{j}\Phi_{k}^{-1}H^{j+1-k}, \:  & \quad  \mathrm{for} \; k > j. \end{array} \right.
 \end{equation}
\end{lem}

\begin{proof} The proof follows by the inspection of  the closed--loop TFMs from the disturbances $w_0$ and $w$ to the  errors $z$ and to the control signals $u$, respectively. The TFM from the disturbances to the errors expressed in terms of the particular doubly coprime factors from (\ref{hurray}) reads (according to Proposition~\ref{cheia1})

\begin{equation} \label{numauna}
z=(\tilde Y_\wp I_n-HN_\wp Q) \tilde M_\wp G_\wp \Phi_0 V_1^n \: \star w_0-(\tilde Y_\wp -HN_\wp Q)  \tilde N_\wp T\Phi \: \star w
\end{equation}

\noindent which implies (\ref{une}) and (\ref{troix}), respectively. 
Furthermore, the TFM from the disturbances $w_0$ and $w$ respectively, to the controls $u$ expressed in terms of the  doubly coprime factors from (\ref{hurray}) reads (according to Proposition~\ref{cheia1})
\begin{equation} \label{numadoua}
u= \Phi^{-1} T^{-1}(\tilde X_\wp +HM_\wp Q) \tilde M_\wp G_\wp \Phi_0 V_1\star  w_0 -   \Phi^{-1} T^{-1}(\tilde X_\wp +HM_\wp Q) \tilde N_\wp T\Phi \star w
\end{equation}

\noindent which in turn yields (\ref{deux}) and (\ref{quatre}), respectively. \end{proof}

\begin{rem} \label{namolbis} As a direct consequence of Lemma~\ref{stringstab}, it follows that under Assumption~\ref{A4} any  leader information controller $K_Q$ also stabilizes $T_{zw_0}$ and $T_{uw_0}$, clarifying the issues raised in Remark~\ref{namol}.
\end{rem}

\begin{rem}
Note that according to (\ref{deux}) the disturbances $w_0$ affecting the leader vehicle,  influence the control signals $u_k$ 
of all other vehicles in the platoon\footnote{The same statement holds true for the leader's controls $u_0$, as well. The leader's controls $u_0$  influence all other control signals $u_k$, with $k \geq 1$.}, since the controls of all followers act to compensate the effect of  $w_0$ on the inter--spacing errors. Interestingly enough, it turns out that this is not necessarily the case for disturbances at the following vehicles. Note that if we take the diagonal Youla parameters in Lemma~\ref{stringstab}  to have identical diagonal entries then the closed--loop TFM $T_{uw}(\la)$ becomes diagonal and consequently the disturbances $w_j$ at the $j$--th vehicle are only ``felt'' on the controls of the $j$--th vehicle $u_j$ and not at all for its successors. 
\end{rem}

We switch now to the second goal of the current section. 
\subsection{Considerations on Local and Global Optimality} \label{consid}
One of the canonical  problems in classical control (dubbed disturbances attenuation) is to design the controller which minimizes some specified norm of the closed--loop TFM   from the disturbances $w$ to the error signals $z$, namely $T_{zw}(\la)$. In the platooning setting, in view of Lemma~\ref{stringstab}, an elementary question one should ask is: what level of disturbances attenuation can be attained by leader information controllers with respect to the \underline{\em local} performance metric $\|T_{z_jw_j}\|$ from (\ref{troix}) at each vehicle ($1\leq j \leq n$ in the platoon). The following result shows that constraining the stabilizing controller to be a leader information controller, does not cause any loss in local  performance, irrespective of the chosen norm (relative to the performance achievable by the centralized optimal controller).

\begin{theorem} \label{AA}
For any $1\leq j \leq n$, the minimum in
\begin{equation} \label{AS7}
\min_{\begin{array}{c} K_Q \; \text{stabilizes} \; G \end{array}}  \; \big \|  T_{z_j w_j}  \big \|
\end{equation}
is attained by a leader information controller. The norm in (\ref{AS7}) can be taken to be either the $\mathcal{H}_2$ or the $\mathcal{H}_\infty$ norm.
\end{theorem}
\begin{proof}  In order to account for any stabilizing controller $K_Q$ in (\ref{AS7})  (possibly centralized controllers), we remove the diagonal constraints on the Youla parameter  from Theorem~\ref{leaderinftare} and  consider generic Youla parameters $Q\in\mathbb{R}(s)^{n \times n}$. Expressing $T_{zw}$ from (\ref{tempo2}) of Proposition~\ref{cheia1} in terms of the doubly coprime factorization (\ref{hurray}) of Theorem~\ref{leaderinftare} yields
\begin{equation} \label{acum}
T_{zw}=-(\tilde Y_\wp -HN_\wp Q)  \tilde N_\wp T\Phi. 
\end{equation}
Note that since $Q$ is no longer assumed to be diagonal, $T_{zw}$ in (\ref{acum}) is no longer lower bidiagonal. Taking (\ref{acum}) into account for the expression of the cost function in (\ref{AS7})
it can be observed that $T_{z_jw_j}$ depends only on the $Q_{jj} , Q_{j(j+1)}$ entries of the Youla parameter, in particular \begin{subequations}\label{tre2}
\begin{equation}
T_{z_jw_j}=-\tilde Y_\wp \tilde N_\wp + N_\wp \big(Q_{jj}-Q_{j(j+1)}\big)\tilde N_\wp, \quad \text{for any} \, 1 \leq j\leq n,
\end{equation}
\begin{equation}
T_{z_jw_{j-1}}=\tilde Y_\wp \tilde N_\wp + N_\wp \big(Q_{j(j-1)}-Q_{jj}\big)\tilde N_\wp, \quad \text{for any} \, 2 \leq j\leq n.
\end{equation}
\end{subequations}
Rewriting (\ref{AS7}) in accordance with (\ref{tre2}), we get
\begin{equation} \label{synthesisAS7bis}
\min_{\begin{array}{c} Q_{jj} , Q_{j(j+1)}  \in \mathbb{R}(s) \\ Q_{jj} , Q_{j(j+1)}  \: \text{stable} \end{array}}  \; \;  \big \| -\tilde Y_\wp \tilde N_\wp H \Phi_j+ HN_\wp \big(Q_{jj}H-Q_{j(j+1)}\big)\tilde N_\wp \Phi_j  \big\|.
\end{equation}
It can be observed that 
if we denote $T_1\overset{def}{=}-\tilde Y_\wp \tilde N_\wp H \Phi_j$ and  $\displaystyle T_2\overset{def}{=}\ba{cc}HN_\wp H\tilde N_\wp \Phi_j   & -HN_\wp \tilde N_\wp \Phi_j  \ea$, with $T_1 \in \mathbb{R}(s)^{}$ and $T_2 \in \mathbb{R}(s)^{1\times 2}$, then (\ref{synthesisAS7bis}) is further  equivalent to
\begin{equation} \label{synthesisAS7bisbis}
\min_{\begin{array}{c} Q_{jj} , Q_{j(j+1)}  \in \mathbb{R}(s) \\ Q_{jj} , Q_{j(j+1)}  \: \text{stable} \end{array}}  \; \;  \big \| T_1+T_2 \ba{c}  Q_{jj} \\Q_{j(j+1)} \ea  \big\|,
\end{equation}
\noindent which  is a standard\footnote{After taking all products, the factors involved in the  model--matching problem  
end up being proper transfer functions. The cause of this is the expression (\ref{Ha}) of the improper $H$ combined with the fact that both $N_\wp$ and $\tilde N_\wp$ are strictly proper (Assumption~\ref{A4}).} model--matching problem which can be solved efficiently for the optimal $Q_{jj} , Q_{j(j+1)}$ \cite{Boyd,Francis}.  Furthermore, it can be observed that if $Q_{jj}^* , Q_{j(j+1)}^*$ is a solution to  (\ref{synthesisAS7bis}) then $\tilde Q_{jj}=Q_{jj}^* - Q_{j(j+1)}^*H^{-1}$, $\tilde Q_{j(j+1)}=0$ is also a solution to (\ref{synthesisAS7bis}). Therefore the minimum can be attained for each one of the $n$ local cost--functions from (\ref{AS7}),  via the diagonal Youla parameter $Q^*\overset{def}{=} \mathcal{D} \{\tilde Q_{11}, \tilde Q_{22}, \dots, \tilde Q_{nn} \}$, which plugged into (\ref{exceptzeonal})  yields the optimal  leader information controller.
\end{proof}

 Interestingly   enough, the following theorem shows that for homogeneous strings of vehicles and constant inter--spacing policies, 
the leader information controller achieves global optimality
(in the $\mathcal{H}_2$ norm), {\em i.e.}, the same performance as the fully centralized controller.
\begin{theorem}
If we assume all vehicles are identical (by taking $\Phi_k=1$, for all $1 \leq k \leq n)$ 
and if we impose constant inter--spacing policies (\ref{spacing}) (by taking the constant time headway\footnote{See also footnote before equation (\ref{spacinghead}).}  $h=0$ in (\ref{Ha}) or equivalently $H(s)=1$), then the optimal leader information controller achieves global $\mathcal{H}_2$ optimality, {\em i.e.}, the minimum in
\begin{equation} \label{AS8}
\min_{\begin{array}{c} K_Q \; \text{\em{stabilizes}} \; G \end{array}}  \; \big \|  T_{zw}  \big \|_2^2
\end{equation}
\noindent is attained by a {\em leader information} controller.
\end{theorem}
\begin{proof}
We will use the following property of the $\mathcal{H}_2$ norm
\begin{equation} \label{tre1}
 \big \|  T_{zw} \big \|_2^2 = \sum_{i=1}^{n} \sum_{j=1}^{n}  \big \| T_{z_iw_j} \big \|_2 ^2
\end{equation}
By taking the lower bi--diagonal terms only, it follows that  (\ref{tre1}) further implies
\begin{equation} \label{tre0}
 \big \|  T_{zw} \big \|_2^2 \geq  \sum_{j=1}^{n} \big \|  T_{z_jw_j} \big \|_2 ^2 + \sum_{j=2}^{n} \big \|  T_{z_jw_{j-1}} \big \|_2 ^2
\end{equation}
In order to account for all (possibly centralized) stabilizing controllers $K_Q$ in (\ref{AS8}), we consider generic (not necessarily diagonal) Youla parameters $Q\in\mathbb{R}(s)^{n \times n}$ in the parameterization of Theorem~\ref{leaderinftare}. It follows that


\begin{equation*} 
\min_{\begin{array}{c} K_Q \; \text{{stabilizes}} \; G \end{array}}  \; \big \|  T_{zw}  \big \|_2^2
\end{equation*}
\begin{equation*}
\overset{(\ref{acum})}{=} \min_{\begin{array}{c} Q \in \mathbb{R}(s)^{n \times n} \\ Q \: \text{stable} \end{array}}  \; \;  \big \| -(\tilde Y_\wp-N_\wp Q)  \tilde N_\wp T  \big\|_2^2 
\end{equation*}
\begin{equation*}
\overset{(\ref{tre0}),(\ref{tre2})}{\geq} \min_{\begin{array}{c} Q \in \mathbb{R}(s)^{n \times n} \\ Q \: \text{stable} \end{array}}  \; \;  \sum_{j=1}^{n} \big \|  -\tilde Y_\wp \tilde N_\wp + N_\wp \big(Q_{jj}-Q_{j(j+1)}\big)\tilde N_\wp \big \|_2 ^2 + \sum_{j=2}^{n} \big \| \tilde Y_\wp \tilde N_\wp + N_\wp \big(Q_{j(j-1)}-Q_{jj}\big)\tilde N_\wp \big \|_2 ^2
\end{equation*}
\begin{equation}
\label{st1}
\overset{}{\geq}  \sum_{j=1}^{n} \min_{\begin{array}{c} Q \in \mathbb{R}(s)^{n \times n} \\ Q \: \text{stable} \end{array}}  \; \;  \big \|  -\tilde Y_\wp \tilde N_\wp + N_\wp \big(Q_{jj}-Q_{j(j+1)}\big)\tilde N_\wp \big \|_2 ^2 + \sum_{j=2}^{n} \min_{\begin{array}{c} Q \in \mathbb{R}(s)^{n \times n} \\ Q \: \text{stable} \end{array}}  \; \;  \big \| \tilde Y_\wp \tilde N_\wp  + N_\wp \big(Q_{j(j-1)}-Q_{jj}\big)\tilde N_\wp  \big \|_2 ^2
\end{equation}
\begin{equation}
\label{st2}
\overset{}{=}  \sum_{j=1}^{n} \min_{\begin{array}{c} Q_{jj} \in \mathbb{R}(s)^{} \\ Q_{jj} \: \text{stable} \end{array}}  \; \;  \big \|  -\tilde Y_\wp \tilde N_\wp + N_\wp Q_{jj}\tilde N_\wp \big \|_2 ^2 + \sum_{j=2}^{n} \min_{\begin{array}{c} Q_{jj} \in \mathbb{R}(s)^{} \\ Q_{jj} \: \text{stable} \end{array}}  \; \;  \big \| \tilde Y_\wp \tilde N_\wp + N_\wp (-Q_{jj})\tilde N_\wp \big \|_2 ^2
\end{equation}
\begin{equation}
\label{st3}
\overset{}{=} n(n-1) \min_{\begin{array}{c} Q_{o} \in \mathbb{R}(s)^{} \\ Q_{o} \: \text{stable} \end{array}}  \; \;  \big \| \tilde Y_\wp \tilde N_\wp  - N_\wp Q_{o}\tilde N_\wp \big \|_2 ^2
\end{equation}
The inequality in (\ref{st1}) is caused by the inter--change of the min with the summation, the equality in (\ref{st2}) follows from the fact that the minimum cost can be achieved by diagonal Youla parameters, while the equality (\ref{st3}) follows from the fact the the  resulted minimization problems are identical.

We solve the last $\mathcal{H}_2$ model--matching problem for $Q_o^*$ (see for example \cite{Boyd}) and it follows that  the minimum in (\ref{AS8}) can be attained via the diagonal Youla parameter $Q^*\overset{def}{=} \mathcal{D} \{ Q_{o}^*,  Q_{o}^*, \dots,  Q_{o}^* \}$, with $Q^*\in \mathbb{R}(s)^{n \times n}$. Finally, when $Q^*$ is plugged into (\ref{exceptzeonal}), it  yields the $\mathcal{H}_2$ optimal {\em leader information controller}.
\end{proof}

We remark that the optimal $\mathcal{H}_2$ leader information controller  \begin{equation} \label{ARS}
\min_{\begin{array}{c} K_Q \; \text{stabilizes} \; G \\  K_Q \; \text{leader information  controller} \end{array}}  \; \Big \|  T_{zw}   \Big \|_2
\end{equation}
can also be computed, since according to (\ref{acum}) and to Theorem~\ref{Youlaaa}, the problem in (\ref{ARS}) is equivalent to the following tractable $\mathcal{H}_2$ model--matching problem \cite{Boyd}

\begin{equation} \label{ARSS}
\min_{\begin{array}{c} Q \in \mathbb{R}(s)^{n \times n} \\  Q \text{diagonal} \end{array}}  \; \Big \| -(\tilde Y_\wp -HN_\wp Q)  \tilde N_\wp T\Phi \Big \|_2
\end{equation}

\subsection{A Practical  $\mathcal{H}_\infty$ Criterion for Controller Design}\label{Praktika}

\underline{The $j$--th Local Problem.} In practice, the local performance objective at the $j$--th vehicle in the platoon ($1 \leq j \leq n$) is formulated as to minimize the effect of the disturbances $w_j$ (at the $j$--th vehicle) both on the interspacing error $z_j$ and on the control effort $u_j$, namely

\begin{equation} \label{AS}
\min_{\begin{array}{c} K_Q \; \text{stabilizes} \; G \\  K_Q \; \text{leader information  controller} \end{array}}  \; \Big \|  \ba{c} T_{z_j w_j} \\ T_{u_j w_j}  \ea  \Big \|_\infty
\end{equation}

\noindent The closed--loop TFM from the disturbances to the controls  is included in the cost in order  to avoid actuator saturation, to regulate the control effort but also to set ``the road attitude'' of the $j$--th vehicle. The $\mathcal{H}_\infty$ norm is used in order to guarantee attenuation  in ``the worst case scenario''. We have dubbed the problem in (\ref{AS}) as {\em the $j$--th local problem}. A convenient feature of the leader information controllers is that both closed--loop terms $T_{z_j w_j}$ and $T_{u_j w_j}$ involved in the cost functional of (\ref{AS}) depend only on the  $Q_{jj}$ entry of the diagonal Youla parameter from Theorem~\ref{leaderinftare}. Therefore, in accordance with Lemma~\ref{stringstab}, when we perform the minimization in (\ref{AS}) after all stabilizing {\em leader information controllers}, 
the $j$--th local  problem (\ref{AS}) can be recast as the following standard\footnote{Due to similar arguments as in footnote (8).}  
$\mathcal{H}_\infty$ model--matching problem \cite[Chapter~8]{Francis}:

\begin{equation} \label{synthesis}
\min_{\begin{array}{c} Q_{jj} \in \mathbb{R}(s) \\ Q_{jj} \: \text{stable} \end{array}}  \; \;  \Big \| \ba{c}- \tilde Y _\wp H \Phi_j \\ - \tilde X_\wp \ea \tilde N_\wp + H\tilde N_\wp \ba{c} N_\wp H\: \Phi_j\\ -M_\wp \ea Q_{jj}  \: \Big\|_\infty. 
\end{equation}


\noindent Note that (\ref{synthesis}) can be efficiently solved for $Q_{jj}$ using existing $\mathcal{H}_\infty$ synthesis numerical routines.  Furthermore, we can always design a {\em leader information controller}  that  simultaneously solves the {\em local problems} for each one of the $n$ vehicles in the string. This is done by solving independently (in parallel, if needed) each $j$--th local problem, for the $j$--th diagonal entry $Q_{jj}$ of the Youla parameter. When plugged into the {\em leader controller} parameterization  of Corollary~\ref{invent}, the resulted diagonal Youla parameter  yields the expression for the local controllers $K_k$ to be placed on board the $k$--th vehicle, ($1\leq k \leq n$).

The local performance objectives imposed in (\ref{AS}) are not sufficient to guarantee the overall behavior of the platoon. The standard  analysis for platooning systems must take into account the effects of the disturbance $w_j$ (at any $j$--th vehicle in the platoon or at the leader) on the errors $z_k$ and controls $u_k$, for all successors in the string ($k>j$).   We will prove next that, as a bonus feature of leader information controllers, the effect of the disturbances $w_j$ on any of its successors $k$ in the string does \underline{not} formally depend on the number $(k-j)$ of  in--between vehicles but only on the following factors: (i) the attenuations obtained at the {\em $j$--th}  and {\em $(j+1)$--th  local problems} respectively (which are optimized by design in (\ref{AS}));  (ii) the stable, minimum phase dynamics $\Phi_j$ and $\Phi_k$ particular to the $j$--th and the $k$--th vehicle, respectively; and (iii) the constant time--headway $H$. In particular, the effect of the disturbances $w_0$ at the leader  on any successor $k\geq 1$ in the string depends on the following: (i) the attenuation obtained at the {\em  $1$--st  local problem}; (ii) the stable, minimum phase dynamics $\Phi_0$ and $\Phi_k$ particular to the leader and the $k$--th vehicle, respectively; and (iii) the constant time--headway $H$.  The precise statement follows:

\begin{coro} \label{stringstabcoro} For any leader information controller, the propagation effect of the disturbances  towards the back of the platoon ({\em sensitivity to disturbances}) is bounded as follows:

  {\bf (A)} The amplification of the disturbance $w_0$ (to the leader's vehicle)  on

   \ \ \ \ \ \  {\bf i)}  the first vehicle in the platoon is given by
\begin{equation}
\Big \|  \ba{c} T_{z_1 w_0} \\ T_{u_1 w_0}  \ea  \Big \|_\infty = \Big \|  \: - \Phi_1^{-1} \Phi_0 H^{-1} \:  \ba{c} T_{z_1 w_1} \\ T_{u_1 w_1}  \ea  \Big \|_\infty  \leq  \Big \|  \: \Phi_1^{-1} \Phi_0 H^{-1} \:  \Big\|_\infty  \Big\| \ba{c} T_{z_1 w_1} \\ T_{u_1 w_1}  \ea  \Big \|_\infty; 
\end{equation}

   \ \ \ \ \ \  {\bf ii)}  the $k$--th vehicle in the platoon, with $k \geq 2$, is given by
\begin{equation}
\Big \|  \ba{c} T_{z_k w_0} \\ T_{u_k w_0}  \ea  \Big \|_\infty = \Big \|  \ba{cc} O & O  \\ O & - \Phi_k^{-1}\Phi_0 H^{-k} \ea  \ba{c} T_{z_1 w_1} \\  T_{u_1 w_1}  \ea  \Big \|_\infty  \leq  \Big \|  \Phi_k^{-1}\Phi_0 H^{-k}  \Big\|_\infty \Big\| \ba{c} T_{z_1 w_1} \\ T_{u_1 w_1}  \ea  \Big \|_\infty. 
\end{equation}

  {\bf (B)} The  amplification of disturbances $w_j$ 

     \ \ \ \ \ \  {\bf i)}  on the $(j+1)$--th vehicle is given by

\begin{equation*}
\Big \|  \ba{c} T_{z_{j+1} w_j} \\ T_{u_{j+1} w_j}  \ea  \Big \|_\infty =  \Big \|  \ba{cc} -H^{-1} & O  \\ O &  \Phi_{j}^{}\Phi_{j+1}^{-1} \ea  \ba{c} T_{z_j w_j} \\  T_{u_j w_j}  \ea  -   \ba{cc} O & O  \\ O &  \Phi_{j}^{}\Phi_{j+1}^{-1} \ea  \ba{c} T_{z_{j+1} w_{j+1}} \\  T_{u_{j+1} w_{j+1}}  \ea \Big \|_\infty  \leq
\end{equation*}

\begin{equation}
\leq  \Big \|  \ba{cc} -H^{-1} & O  \\ O &  \Phi_{j}^{}\Phi_{j+1}^{-1} \ea \Big \|_\infty   \Big \|   \ba{c} T_{z_j w_j} \\  T_{u_j w_j}  \ea \Big \|_\infty   + \Big \|    \ba{cc} O & O  \\ O &  \Phi_{j}^{}\Phi_{j+1}^{-1} \ea \Big \|_\infty   \Big \|   \ba{c} T_{z_{j+1} w_{j+1}} \\  T_{u_{j+1} w_{j+1}}  \ea \Big \|_\infty;
\end{equation}

   \ \ \ \ \ \  {\bf ii)}  on the $k$--th vehicle, with $k \geq j+2$, is given by

   \begin{equation*}
\Big \|  \ba{c} T_{z_{k} w_j} \\ T_{u_{k} w_j}  \ea  \Big \|_\infty =  \Big \|  \ba{cc} O & O  \\ O &  \Phi_{j}^{}\Phi_{k}^{-1} H^{j+1-k} \ea  \ba{c} T_{z_j w_j} \\  T_{u_j w_j}  \ea  -   \ba{cc} O & O  \\ O &  \Phi_{j}^{}\Phi_{k}^{-1}H^{j+1-k} \ea  \ba{c} T_{z_{j+1} w_{j+1}} \\  T_{u_{j+1} w_{j+1}}  \ea \Big \|_\infty  \leq
\end{equation*}

\begin{equation}
\leq  \Big \|   \Phi_{j}^{}\Phi_{k}^{-1} H^{j+1-k} \Big \|_\infty   \Big \|   \ba{c} T_{z_j w_j} \\  T_{u_j w_j}  \ea \Big \|_\infty   + \Big \|   \Phi_{j}^{}\Phi_{k}^{-1} H^{j+1-k} \Big \|_\infty  \Big \|   \ba{c} T_{z_{j+1} w_{j+1}} \\  T_{u_{j+1} w_{j+1}}  \ea \Big \|_\infty. 
\end{equation}

\end{coro}

\begin{proof}
The proof follows by straightforward algebraic manipulations of the expressions of the closed--loop TFMs provided by Lemma~\ref{stringstab}  and by the sub--multiplicative property of the $\mathcal{H}_\infty$ norm.
\end{proof}

\begin{rem} It is important to remark here that if we are to consider constant inter-spacing policies (or, equivalently, if we take the expression of  the constant time--headway $H(s)=1$),  then the attenuation bounds provided  by Theorem~\ref{stringstab} do not depend on the number $(k-j)$  of  in-between vehicles.  This is consistent with the definition  introduced in \cite{rick} for {\em string stability} of platoon formations. Furthermore, if we do consider constant time--headway policies then the negative powers of the constant time--headway $H(s)$ having subunitary norm, will introduce additional attenuation, especially at high frequency via the strong effect of the roll-off.
\end{rem}

\begin{rem} \label{meri} Vehicles desiring to enter the formation  should indicate their intention to the vehicles in the string. The vehicles in the string where the merging maneuver is to be performed may increase their interspacing distance  ({\em e.g.} the distance based headway component of the interspacing policy) such as to allow for the merging vehicle to enter the formation safely.  A remarkable feature of the leader information controllers introduced here is the fact that when dealing with merging traffic the only needed reconfiguration of the global scheme is at the follower of the merging vehicle, which must acknowledge the ``new'' unimodular factor $\Phi_k$ of the vehicle appearing in front of it. Such a maneuver can be looked at as a disturbance to the merging vehicle, to be quickly attenuated by the control scheme. Equally important, if the broadcasting of any vehicle in the platoon gets disrupted, then the global scheme can easily reconfigure, such that the non--broadcasting vehicle becomes the leader of a new platoon.
\end{rem}

\section{Dealing with Communications Induced Time--Delays}
\label{delaydelay}

In this section we look at the factual scenario when there exists a time delay on each of  the feedforward links $u_k$, with $1 \leq k \leq (n-1)$. In practice, these delays are caused by the  physical limitations of the wireless communications system used for the implementation of the feedforward link, entailing a time delay $e^{-\theta s}$ (with $\theta$ typically around $20$ ms\footnote{For wireless communications systems based on high frequency digital radio, such as WiFi, ZigBee or Bluetooth. In practice, the time delays will be time--varying, but they can be well--approximated by a constant of their corresponding {\em nominal value}.}) at the receiver of the broadcasted  $u_k$ signal (with $1 \leq k \leq  (n-1)$). We consider that the delay is the same for all vehicles, since all members of the platoon use similar wireless communications systems and we assume that the delay is known from technological specifications. This type of situation is represented  in Figure~\ref{f2bis}, if we consider the  $\Psi$ blocks to be equal to $e^{-\theta s}$, with $\theta \neq 0$.  We will show next  that in the presence of such time delays in the implementation of the  leader information controllers of Theorem \ref{leaderinftare} (and Corollary~\ref{invent}), the diagonal sparsity pattern of the resulted closed--loop TFM $T_{zw_0}$ is compromised as it becomes lower triangular and it no  longer satisfies Definition~\ref{leader}. This means that the resulted (wireless communications based) physical implementation  of any 
controller from Corollary~\ref{invent} will in fact not be a {\em bona fide}  leader information controller.  Furthermore, it can be shown that the effects of the communications delays drastically alter the closed loop performance \cite{rick} as they necessarily lead to string instability.


\subsection{The Effect of Communications Time Delays on the Control Performance}
In order to make our point with illustrative simplicity let us consider (for this subsection only) the case of platoons with identical vehicles  ({\em i.e.}, $\Phi_k=1$, for all $1\leq k \leq n$) and constant interspacing policies, ({\em i.e.},  $H(s)=1$). Under these assumptions, the equation of the controller from Figure~\ref{f2bis}, with $\Psi=e^{-\theta s}$, reads:
\begin{equation}
 \ba{c} u_1 \\ u_2 \\ u_3 \\ \vdots \\ u_{n-1} \\ u_n  \ea = \ba{ccccccc}  O & O & O &  \dots & O & O\\ e^{-\theta s} & O & O &  \dots & O & O\\ O & e^{-\theta s} & O &  \dots & O & O \\ O & O & e^{-\theta s} &  \dots & O & O \\ \vdots & \vdots & \vdots & \ddots & \vdots & \vdots \\  O & O & O &  \dots & e^{-\theta_{} s}  & O  \ea \star \ba{c} u_1 \\ u_2 \\ u_3 \\ \vdots \\ u_{n-1} \\ u_n\ea +
 \ba{cccccc}  K_{1} & O & O & \dots  & O & O\\ O & K_{2}  & O & \dots  & O & O \\ O & O & K_{3} & \dots  & O & O \\ O & O & O & \ddots & O & O  \\ \vdots & \vdots & \vdots & \ddots & \vdots & O \\ O & O & O &  \dots & O & K_n \\  \ea \star \ba{c} z_1 \\ z_2 \\ z_3 \\ \vdots \\ z_{n-1} \\ z_n \ea
 \end{equation}

\noindent  or equivalently

\begin{equation}  \label{residual}
\ba{c} u_1 \\  u_2 \\ u_3 \\ \vdots \\  u_{n-1} \\ u_n  \ea = \ba{ccccccc}  1 & O & O &  \dots & O & O\\ e^{-\theta s} & 1 & O &  \dots & O & O\\ e^{-2\theta s} & e^{-\theta s} & 1 &  \dots & O & O \\ e^{-3 \theta s} & e^{-2 \theta s} &  e^{-\theta s}  &  \ddots & O & O \\ \vdots & \vdots & \vdots & \ddots & \vdots & \vdots \\  e^{(-n+1) \: \theta s} & e^{(-n+2) \: \theta s}  & e^{(-n+3) \: \theta s}&  \dots & e^{- \theta s}  & 1  \ea \ba{cccccc}  K_{1} & O & O & \dots  & O & O\\ O & K_{2} & O & \dots  & O & O \\ O & O & K_{3}  & \dots  & O & O \\ O & O & O & \ddots & O & O  \\ \vdots & \vdots & \vdots & \ddots & \vdots & O \\  O & O&  O &  \dots & O & K_n \\  \ea \star \ba{c} z_1 \\  z_2 \\  z_3 \\ \vdots \\  z_{n-1} \\  z_n \ea
\end{equation}

By employing Theorem~\ref{leaderinftare}, it can be checked that {\em any} leader information controller,  belongs to the following set $\mathcal{S}$, defined as
\begin{equation} \label{SSSS}
\mathcal{S}\overset{def}{=} \big\{ K\in \mathbb{R}(s)^{n \times n} \:  \big | K=T^{-1}\mathcal{D} \{D_{11}, D_{22} \dots D_{nn}\} \: \text{with} \: D_{jj}\in \mathbb{R}(s), \: 1\leq j\leq n\big\}.
\end{equation}
The argument follows by straight forward algebraic manipulations starting from the right coprime factorization $K_Q=\tilde X_Q \tilde Y_Q^{-1}$ of any leader information controller.  Clearly, the controller from (\ref{residual}) belongs to the set $\mathcal{S}$ in (\ref{SSSS}) (and is therefore a  leader information controller) if and only if $\theta=0$ or, equivalently, in the absence of any communications delay. We also remark from (\ref{residual}) that the time--delays  propagate  ``through the controller'' downstream the platoon  and the delays accumulate toward the end of the platoon, in a manner depending on the number of vehicles  in the string (specifically $n$). An in depth analysis of the propagation effect of feedforward communications delays 
through a  platoon of vehicles can be found in \cite{rick}.
\begin{figure}
\hspace{-3mm}
\includegraphics[scale=.2]{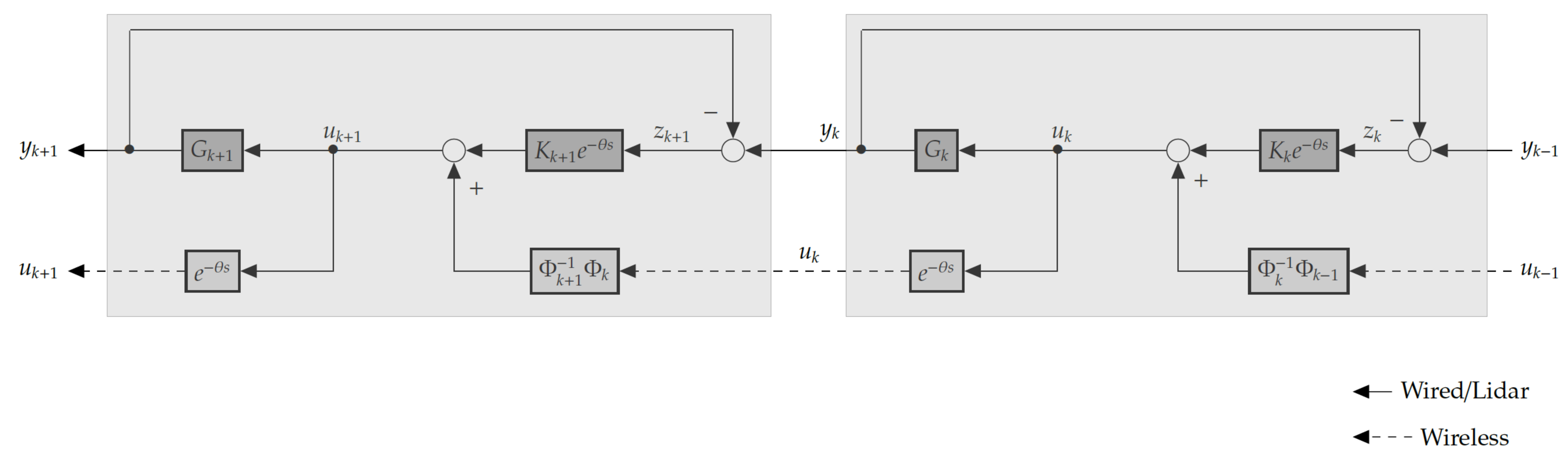}
\caption{Leader Information Control with Compensation of Communications Delay}
\label{f3}
\end{figure} 

\begin{rem}\label{haihai}
We remind here the basic fact known in control theory that the delay $e^{-\theta s}$ (on any of the feedforward channels $u_k$,  with $1\leq k \leq n$) cannot be efficiently compensated by a series connection with a linear filter on the feedforward path, such as a rational function approximation of the anticipative element $e^{\theta s}$, that would ``cancel out'' the effect of the delay.
\end{rem}

\subsection{A Delay Compensation Mechanism Using  Synchronization}

 In this subsection we will show how the communications induced delays can be compensated at the expense of a negligible loss in performance.  We place a delay of  exactly $\theta$ seconds on each of the sensor measurements $z_k$. This delay appears in Figure~\ref{f3}\footnote{To make the graphics more readable we  have illustrated  the case in which the constant time--headway policy  has been removed, meaning that we considered $H(s) =1$. See also the footnote related to Figure~\ref{f2bis}.} 
 as an $e^{-\theta s}$ factor in the transfer function  $K_{k+1}$, for any $1\leq k \leq n-1$.

Having a delay $e^{-\theta s}$ on both $u_k$   and $z_{k+1}$ is equivalent with having an $e^{-\theta s}$  delay in the model of the $(k+1)$--th vehicle  $G_{k+1}$, for any $1\leq k \leq n-1$.   The argument for this fact is the following controller equation for the equivalent scheme of Figure~\ref{f4}:

\begin{equation*}
\ba{c} u_1 \\ e^{-\theta s}\star u_2 \\ e^{-\theta s}\star u_3 \\ \vdots \\ e^{-\theta s}\star u_{n-1} \\ e^{-\theta s}\star u_n  \ea = H^{-1}\ba{ccccccc}  O & O & O &  \dots & O & O\\ \Phi_2^{-1}\Phi_1e^{-\theta s} & O & O &  \dots & O & O\\ O & \Phi_3^{-1}\Phi_2e^{-\theta s} & O &  \dots & O & O \\ O & O & \Phi_4^{-1}\Phi_3 e^{-\theta s} &  \dots & O & O \\ \vdots & \vdots & \vdots & \ddots & \vdots & \vdots \\  O & O & O &  \dots &\Phi_n^{-1}\Phi_{n-1} e^{-\theta s}  & O  \ea \star \ba{c} u_1 \\ u_2 \\ u_3 \\ \vdots \\ u_{n-1} \\ u_n\ea+
\end{equation*}
\begin{equation} \label{nice}
+ H^{-1}\ba{cccccc}  K_{1} e^{-\theta s} & O & O & \dots  & O & O\\ O & K_{2} e^{-\theta s}\: & O & \dots  & O & O \\ O & O & K_{3} e^{-\theta s}\: & \dots  & O & O \\ O & O & O & \ddots & O & O  \\ \vdots & \vdots & \vdots & \ddots & \vdots & O \\ O & O & O &  \dots & O & K_n e^{-\theta s}\: \\  \ea \star \ba{c} z_1 \\  z_2 \\  z_3 \\ \vdots \\  z_{n-1} \\  z_n \ea
\end{equation}

Note that (\ref{nice}) results directly from the equation for the leader information controller of Corollary~\ref{invent} by multiplying both sides  to the left with the $n \times n$ diagonal TFM $\mathcal{D} \{1, e^{-\theta s}\:, e^{-\theta s}\:, \dots, e^{-\theta_{} s}\:\}$. It is important to  observe that  the controller given in (\ref{nice}) acts on the $e^{-\theta s}$ delayed version of the platoon's plant 
that the controller given in (\ref{fff}) acts on -- in the statement of Corollary~\ref{invent}. For the purposes of designing the sub--controller $K_{k+1}$, the $e^{-\theta s}$  time delay will be considered to be part of the $G_{k+1}$ plant model  (when employing for instance the methods introduced in Subsection~\ref{Praktika}). 
\begin{rem}
In practice the time delay $\theta$ may be chosen to be the maximum of the latencies of all vehicles in the string, where a vehicle's latency is defined to be the sum of the nominal (or worst case scenario) time delay of the electro--hydraulic actuators  with the nominal (or worst case scenario) time delay of the wireless communications. The homogeneity of the latencies of all vehicles in the string can be simulated and implemented using high accuracy GPS time base synchronization mechanisms. Such synchronization mechanisms will therefore produce fixed, commensurate and point--wise delays, thus avoiding the inherent difficulties caused by time--varying or stochastic  or distributed delays. The LTI controller synthesis can then be performed by taking a conveniently  chosen Pade rational approximation of $e^{-\theta s}$ to be included in the expression of $G_\wp$ from Assumption~\ref{A4}. It is a well known fact that such an approximation will introduce additional non--minimum phase zeros in $G_\wp$ and consequently some loss in performance. However, and this is important, the resulted controllers of Figure~\ref{f4} are leader information controllers and will therefore feature all the structural properties discussed in Sections~\ref{MR} and \ref{tuneup}.
\end{rem}

\begin{figure}
\centering
\includegraphics[scale=.2]{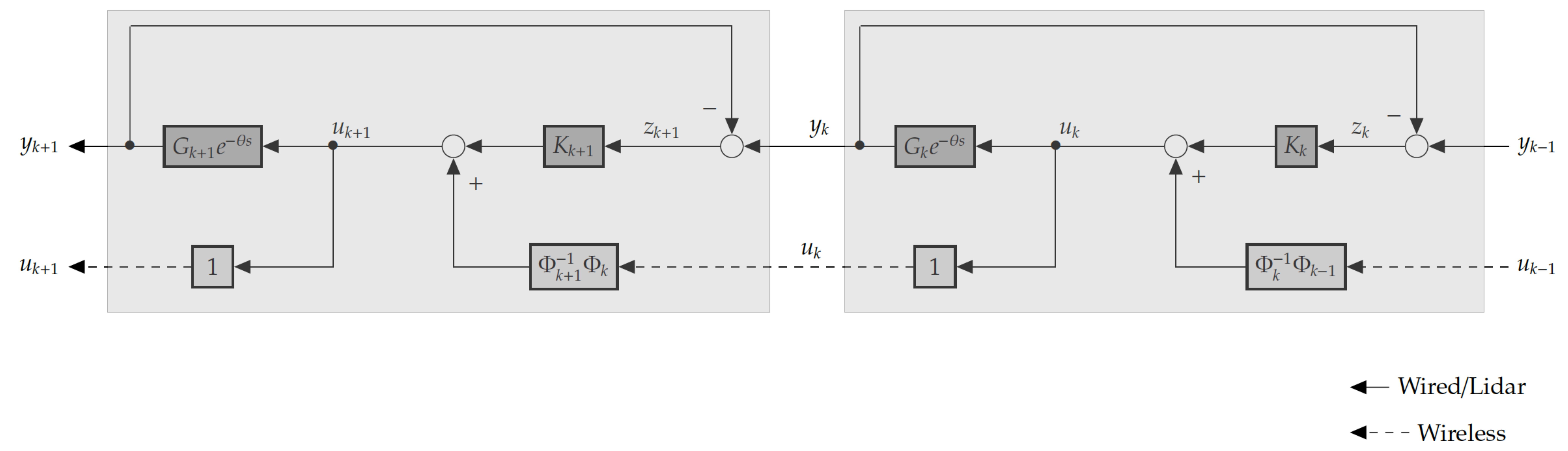}
\caption{Leader Information Control Scheme Equivalent with Figure~\ref{f3}}
\label{f4}
\end{figure}

\section{Comparison with State--of--the--Art}
\label{soa}
In this Section we provide a comparison of our method with existing results. In Subsection~\ref{nnasol} we look at the recent CACC design from \cite{nasol}, which has important conceptual similarities with  leader information controllers. In Subsection~\ref{vrajeala} we look mainly at  the {\em indirect leader brodacast} architectures analyzed in \cite{rick}. Finally, in Subsection~\ref{QQII} we discuss connections with {\em quadratic invariant} feedback configurations.

\subsection{Recent CACC Design Methods} \label{nnasol}

Recently, the authors of \cite{nasol, nasol2} introduced  a control scheme in which each vehicle in the platoon broadcasts its control signal to its successor in the string, in a  similar manner with our leader information controller.  The control law in \cite {nasol} is designed such as to account for  a Pade approximation of the feedforward time--delay induced by the wireless broadcast of the control signal. For this reason, but also due to the manner in which the $\mathcal{H}_\infty$ controller synthesis problem is posed in \cite{nasol, nasol2}, the resulting controller from \cite[Section~V]{nasol}  will never be   a  leader information controller, as we show next.

 {\bf Fact.} The controller having the expression from \cite[(27)/pp.858]{nasol}, which minimizes the $\mathcal{H}_\infty$  mixed sensitivity criterion in \cite[(28)/pp.858]{nasol}  \underline{is not} a {\em leader information controller}.
 \begin{proof} In \cite{nasol} all vehicles are assumed to be dynamically identical, therefore we will take the unimodular  factors $\Phi_k$ from Assumption~\ref{A4} to be the same for all vehicles. Consequently, for any fixed $k \geq 2$, the leader information control law (\ref{fff}) produced on board the $k$--th vehicle, takes the following form (according to Corollary~\ref{invent}):
 \begin{equation} \label{cenasol}
 u_k=u_{k-1} + K_k z_k ,
 \end{equation}
 \noindent with $K_k$ as in (\ref{kakaka}). 
We remark that for {\em any} leader information control law  (\ref{cenasol}), the feedforward filter associated with $u_{k-1}$ in (\ref{fff}) is such that $\Phi_k^{-1}\Phi_{k-1}=1$, which is never the case for the controllers from \cite[(27)/pp.858]{nasol}. The reason for this is  the ``asymmetry'' from \cite{nasol}  between the feedforward branch of $u_{k-1}$ which is time delayed and the feedback branch $z_k$ which is not.
 \end{proof}

The qualitative differences between the two schemes are further illustrated by the wave forms shown in the numerical example provided in the next Section. The numerical example features the structural properties  emphasized in Subsection~\ref{strruct}: it achieves string stability and it renders evident the elimination of the accordion effect {\em in the presence of communications delays}.  This is due to the fact that the  approach in \cite{nasol} only looks at the ``local'' closed--loops associated with a single vehicle in the string, while our analysis examines the closed--loop TFMs  of the entire platoon. Our discussion also concludes that for platooning control  the only  ``local'' measurements needed at each agent in the string are the inter--spacing distance with respect to its predecessor and the predecessor's control signal. This is an important point, since it clarifies previous conjectures \cite[Section~V--B]{nasol},\cite[pp.~5]{TRB}, \cite{nasol2} that additional  information from multiple predecessors (``beyond the direct line of sight'') might lead to superior performance, since they provide a ``preview of disturbances''.

With  respect to the first experimental validation in \cite[Section VI]{nasol} performed for a string comprised of two vehicles, it is worthwhile to mention that the results presented here emphasize the fact that the vehicle immediately following the leader (specifically vehicle with index ``$1$'' in our notation) does not benefit from the transmission of the leader's control signal $u_0$ which is  in general considered to be a reference signal for the entire platoon. This observation is especially useful since it implies that a platoon of vehicles equipped with the current control architecture could follow on the highway a leader vehicle operated by a human driver. Similarly, if the wireless transmission of any vehicle in the platoon gets disrupted, the global control scheme scheme can easily reconfigure such that the non--transmitting vehicle becomes the leader of a new platoon.

\subsection{Other Considerations} \label{vrajeala}
 The so called  {\em indirect leader broadcast} scheme from \cite{rick} studied for homogeneous strings of vehicles presents certain similarities with our leader information controller from Theorem~\ref{leaderinftare},  with the distinct feature that in our leader information controller  we broadcast the control signal of the predecessor vehicle instead of an estimate of the leader's state. The control signal is basically generated on board of the predecessor vehicle, hence there is no need to estimate it and the fact that {\em exact} information is broadcasted  (with some unavoidable time--delay)  has profound implications in terms of the performance of the closed loop. Furthermore, the leader information control scheme from Corollary~\ref{invent}   can be adapted such as to compensate for the feedforward time--delay induced by the wireless communication broadcasting of the predecessor's control signal, as  explained in full detail in  Section~\ref{delaydelay}.

The particular type of structure featured by the controller in (\ref{fff}) has been initially investigated in
\cite{R, CDC} on the basis of the so--called {\em dynamical structure function} of a LTI network, as introduced in \cite{Sean08}. One particular topology discussed in \cite{R, CDC} is the ``ring'' network with LTI dynamics, while the controller from (\ref{fff}) of this paper features a ``line'' topology (in fact a unidirectional  ``ring'' with the link between  agents $n$ and $1$ cut off).  The scope of the state--space analysis from \cite{R, CDC} is to establish the connections between all the left coprime factorizations  (\ref{exceptzeonal}) associated with a certain TFM $K$  and  all possible dynamical structure functions \cite{Sean08} associated with the same TFM.

\subsection{Connections with Quadratic Invariance} \label{QQII}
The $\mathcal{H}_2$ optimality feature discussed in Subsection~\ref{consid}, stimulated the investigation of eventual connections of leader information controllers with quadratic invariant feedback structures.  The so--called quadratically invariant ({\bf QI}) configurations \cite{Michael's_I3E} constitute the largest known class of tractable problems in decentralized control. In this subsection we address the connections between QI and the {\em leader information controllers} for platooning. In many cases of interest, the decentralized nature of the control problem can be formulated by constraining the stabilizing controller $K\in \mathbb{R}(s)^{n \times n}$ to belong to a pre--specified linear subspace $\mathcal{S}$ of $\mathbb{R}(s)^{n \times n}$.  Often, this framework is used to impose sparsity constraints on the controller, by taking for instance $\mathcal{S}$ to be the subspace of all diagonal TFMs  in $\mathbb{R}(s)^{n \times n}$(or the subspace of all lower triangular TFMs in $\mathbb{R}(s)^{n \times n}$). The authors of \cite{Michael's_I3E} identified a property (dubbed quadratic invariance)  of the plant $G$ in conjunction with the controller's constraints set $\mathcal{S}$,  that guarantees a convex parameterization of all admissible stabilizing controllers (belonging to $\mathcal{S}$).

\begin{defn}
\label{QI} \cite[Definition 2]{Michael's_I3E}
A closed linear subspace $\mathcal{S}$ of $\mathbb{R}(s)^{n \times n}$ is called {\em quadratically invariant}  under the plant $G$ if
\begin{equation} \label{EqQI}
 K  G K \in \mathcal{S}, \quad \text{for all} \quad  K\:  \in \mathcal{S}.
\end{equation}
\end{defn}

\begin{defn} \label{fdbckptrans}
Define the  feedback transformation $\varg_G: \mathbb{R}(s)^{n \times n} \rightarrow \mathbb{R}(s)^{n \times n}$ of $G$ with $K$, as follows:
\begin{equation}
\label{hG}
\varg_G(K) \overset{def}{=}K \big{(}I+ G K \big{)}^{-1}, \qquad K \in  \mathbb{R}(s)^{n \times n}.
\end{equation}
\end{defn}




The intrinsic features of QI configurations are rooted in invariance principles (such as the earlier concept of {\em funnel causality} \cite{Voulg1,Voulg2,Voulg3,Voulg4}) best encapsulated by the following property:
\begin{theorem} \cite[Theorem~14]{Michael's_I3E}
\label{Mike's}
Given a sparsity constraint $\mathcal{S}$, the following equivalence holds:
\begin{equation} \label{Rotk}
\mathcal{S}\; \text{is QI under} \; G \Longleftrightarrow \varg_{G}(\mathcal{S} ) = \mathcal{S},
\end{equation} where we adopt the following abuse of notation: $$ \varg_{G}(\mathcal{S} )\overset{def}{=} \{\varg_G(K)\: | \: K \in \mathcal{S} \}.$$
\end{theorem}

The main attribute of QI feedback configurations is that the corresponding constrained optimal $\mathcal{H}_2$--control problem (involving the norm of a Linear Fractional Transformation of the plant $G$) is tractable:
\begin{equation} \label{QILFT}
\min_{\begin{array}{c} K \; \text{stabilizes} \; G \\  K \in \mathcal{S} \end{array}}  \; \Big \|   T_{11} + T_{12}K \big{(}I+ G K \big{)}^{-1}T_{21}    \Big \|_2.
\end{equation}
In (\ref{QILFT}) above, $T_{11}, T_{12}$,$T_{21}$  and $G$ respectively, represent the pre--specified TFMs of a given {\em generalized plant} \cite[Chapter~3]{Francis}. The tractability of (\ref{QILFT}) hinges on the fact that it can always be recast as a $\mathcal{H}_2$ model--matching problem \cite{Boyd} with additional subspace constraints on the Youla parameter \cite[Section~IV--D]{Michael's_I3E},\cite{me}. We are now ready to state the following result, which is the scope of the current subsection:
\begin{prop} \label{Subspace} Given the platoon's  plant $G$ (having the expression given in Assumption~\ref{A4}), let us define
\begin{equation} \label{SS}
\mathcal{S}\overset{def}{=} \big\{ K\in \mathbb{R}(s)^{n \times n} \:  \big | K=\Phi^{-1}T^{-1}\mathcal{D} \{D_{11}, D_{22} \dots D_{nn}\} \: \text{with} \: D_{jj}\in \mathbb{R}(s), \: 1\leq j\leq n\big\}.
\end{equation}
The set $\mathcal{S}$ is a closed linear subspace of $\mathbb{R}(s)^{n \times n}$ having dimension $n$. Furthermore, any leader information controller $K_Q$ belongs to $\mathcal{S}$ and $\mathcal{S}$ is QI under the platoon's plant $G$.
\end{prop}
\begin{proof} Clearly $\mathcal{S}$ is closed under addition and under multiplication with scalar rational functions  in $\mathbb{R}(s)$ and is therefore a linear subspace. One basis of $\mathcal{S}$ is comprised of exactly $n$ TFMs from $\mathbb{R}(s)^{n \times n}$, where (for $1 \leq j \leq n$) the $j$--th TFM in the basis has its $j$--th column identical to the $j$--th column of $\Phi^{-1}T^{-1}$ and zero entries elsewhere. If we apply Theorem~\ref{Youlaaa} to the doubly coprime factorization (\ref{hurray}) of Theorem~\ref{leaderinftare}, then any leader information controller $K_Q$ is of the form $K_Q=\tilde X_Q \tilde Y_Q^{-1}$, where $\displaystyle Q\overset{def}{=}\mathcal{D}\big\{Q_{11},Q_{22}, \dots, Q_{nn}\big\}$ 
is a diagonal Youla parameter. More explicitly, following (\ref{hurray}) any such $K_Q$ can be written as
\begin{equation} \label{SS1}
K_Q=\Phi^{-1}T^{-1} (\tilde X_\wp I_n+HM_\wp Q) (\tilde Y_\wp I_n-HN_\wp Q)^{-1}
\end{equation}
which obviously lies in $\mathcal{S}$, since the involved  Youla parameters $Q$ are diagonal. Finally, we will prove that $\mathcal{S}$ satisfies  Definition~\ref{QI}, with respect to our plant  $G=T \Phi G_\wp$ (from the statement of Theorem~\ref{leaderinftare} ) where $G_\wp  \in \mathbb{R}(s)$. According to (\ref{SS}), for any $K\in\mathcal{S}$ there exists a diagonal TFM $D$ belonging to $\mathbb{R}(s)^{n \times n}$, such that $K=\Phi^{-1}T^{-1}D$. Then $KGK=(\Phi^{-1}T^{-1}D)(T \Phi G_\wp )(\Phi^{-1}T^{-1}D)$. Since $G_\wp$ is a scalar TFM, its multiplication is commutative and  we obtain $KGK=(\Phi^{-1}T^{-1}D)(T \Phi)(\Phi^{-1}T^{-1}D)G_\wp$ and after simplification $KGK=\Phi^{-1}T^{-1}D^2G_\wp$ which belongs to $\mathcal{S}$. The proof ends.
\end{proof}

\begin{rem} \label{cntrb} Previously known practical interpretations for subspace constraints 
consist of the following: sparsity constraints on the controller\footnote{For a practical interpretation of QI {\em sparsity} constraints in terms of the interconnection structure of the distributed controller, we refer to \cite{Nuno}.}, controllers having symmetric TFMs and modeling the communications time-delays between sub--controllers, respectively. We remark that the subspace $\mathcal{S}$ we have introduced in (\ref{SS}) delineates a distinct type of  subspace constraints, which are \underline{not} of the  sparsity type. This is because {\em leader information controllers} are not simply constrained to have lower triangular TFMs. (The subspace of lower triangular TFMs in $\mathbb{R}(s)^{n \times n}$ has dimension $n(n+1)/2$, while the $\mathcal{S}$ subspace from (\ref{SS}) has dimension $n$). It is especially noteworthy that the particular structure enforced by $\mathcal{S}$ on the {\em leader information controllers} is not relevant in itself to a distributed implementation of the controller, such as the particularly useful one from Corollary~\ref{invent}. In turn, the meaningful structure of {\em leader information controllers} is  completely captured by the sparsity constraints imposed on their left coprime factors, as specified in  Theorem~\ref{leaderinftare}.
\end{rem}

For the platooning problem, the QI specific  type (\ref{QILFT}) cost  is involved  in the expression of $T_{uw_0}=K(I_GK)^{-1}V_1G_\wp \Phi_0$ (Proposition~\ref{cheia1} (B)). Accordingly, a direct consequence of Proposition~\ref{Subspace} is the tractability of the minimization problem 
of the control effort caused by disturbances to the leader:

\begin{equation} \label{AS97}
\min_{\begin{array}{c} K_Q \; \text{stabilizes} \; G \\  K_Q \; \text{leader information  controller} \end{array}}  \; \big \|   T_{uw_0} \big \|_2. 
\end{equation}

More recently, various solutions for the $\mathcal{H}_\infty$ counterpart of the control problem (\ref{QILFT}) for QI configurations have been proposed in \cite{Hinf1,Hinf2}. However, these methods can only cope with the situation when $\mathcal{S}$ is described by sparsity constraints (mainly lower triangular sparsity constraints), therefore they cannot be directly adapted for the {\em leader information controller} constraints of (\ref{SS}).


\section{A numerical example}
\label{nex}
We present in this section a numerical MATLAB simulation for the platoon motion with $n=6$ vehicles, having the transfer function
\begin{equation}\label{mai12}
G_k(s) = \frac{s+\sigma_k}{m_ks^2(\tau_k s+1)}\ e^{-(\phi+\theta)s},\ k=1,2,\dots 6.
\end{equation}
where $\phi = 0.1$ sec. is the electro--hydraulic break/throttle actuator delay and $\theta = 0.03$ sec. is the wireless communications delay. For the $k$-th vehicle, the mass $m_k$, the actuator time constant $\tau_k$ and the stable zero $\sigma_k>0$ are given in Table \ref{param}, next.

\begin{table}[h]
\centering
\begin{tabular}{||c|| c| c| c| c| c| c||}
\hline
k  & 1 & 2 & 3& 4 & 5 & 6\\
\hline
$m_k$ [kg] &8& 4& 1& 3& 2& 7\\
\hline
 $\tau_k$ [s] & 0.1 &    0.2&   0.05&    0.1&    0.1&    0.3\\
\hline
$\sigma_k$ &1 &2& 3& 4& 5& 6\\
\hline
\end{tabular}
\caption{Numerical parameters for the vehicles}
\label{param}
\end{table}

\begin{figure}
\hspace{-12mm}
\begin{tabular}{cc}
\centering
\includegraphics[scale=0.48]{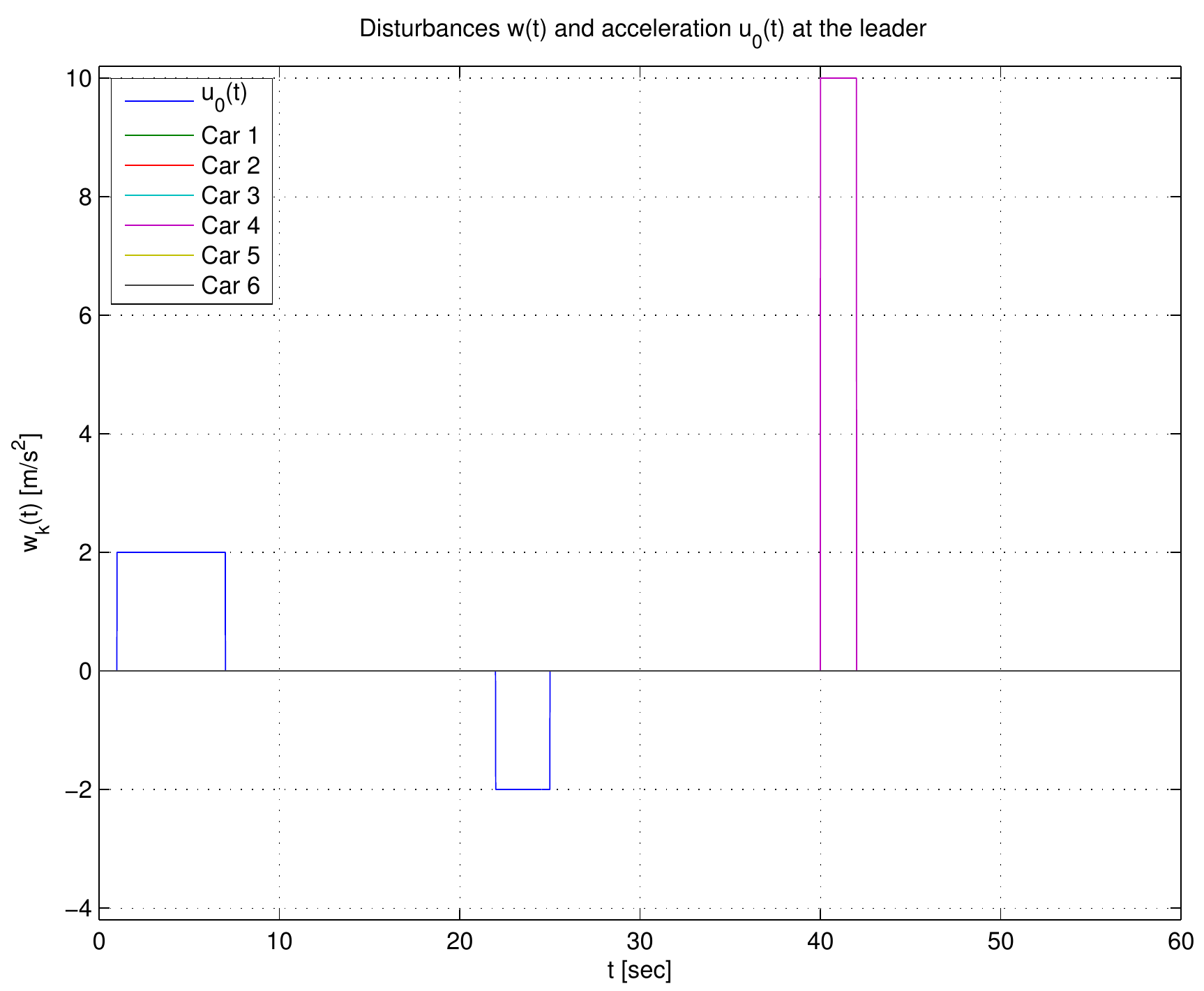}&
\includegraphics[scale=0.48]{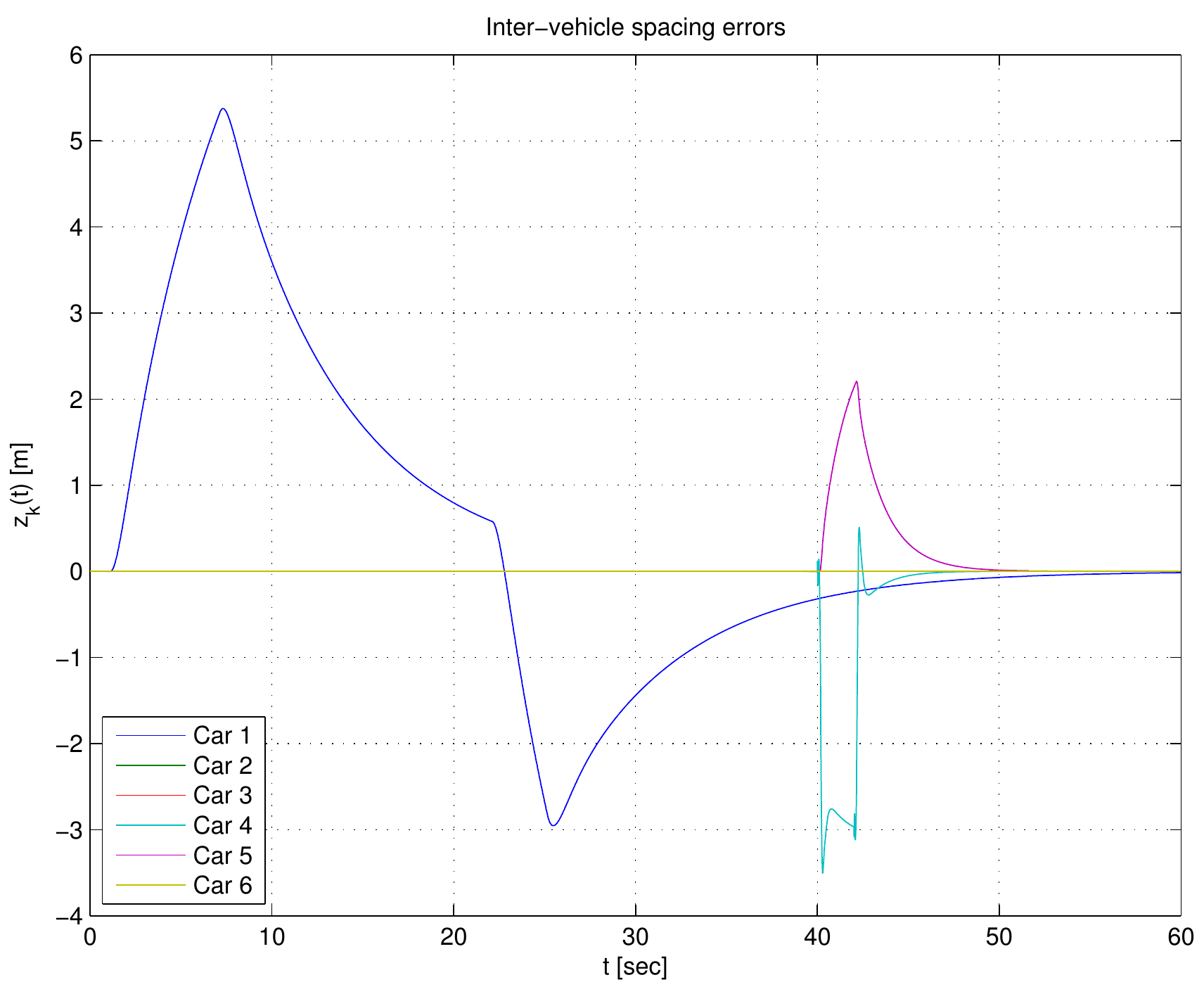}
\\
\includegraphics[scale=0.48]{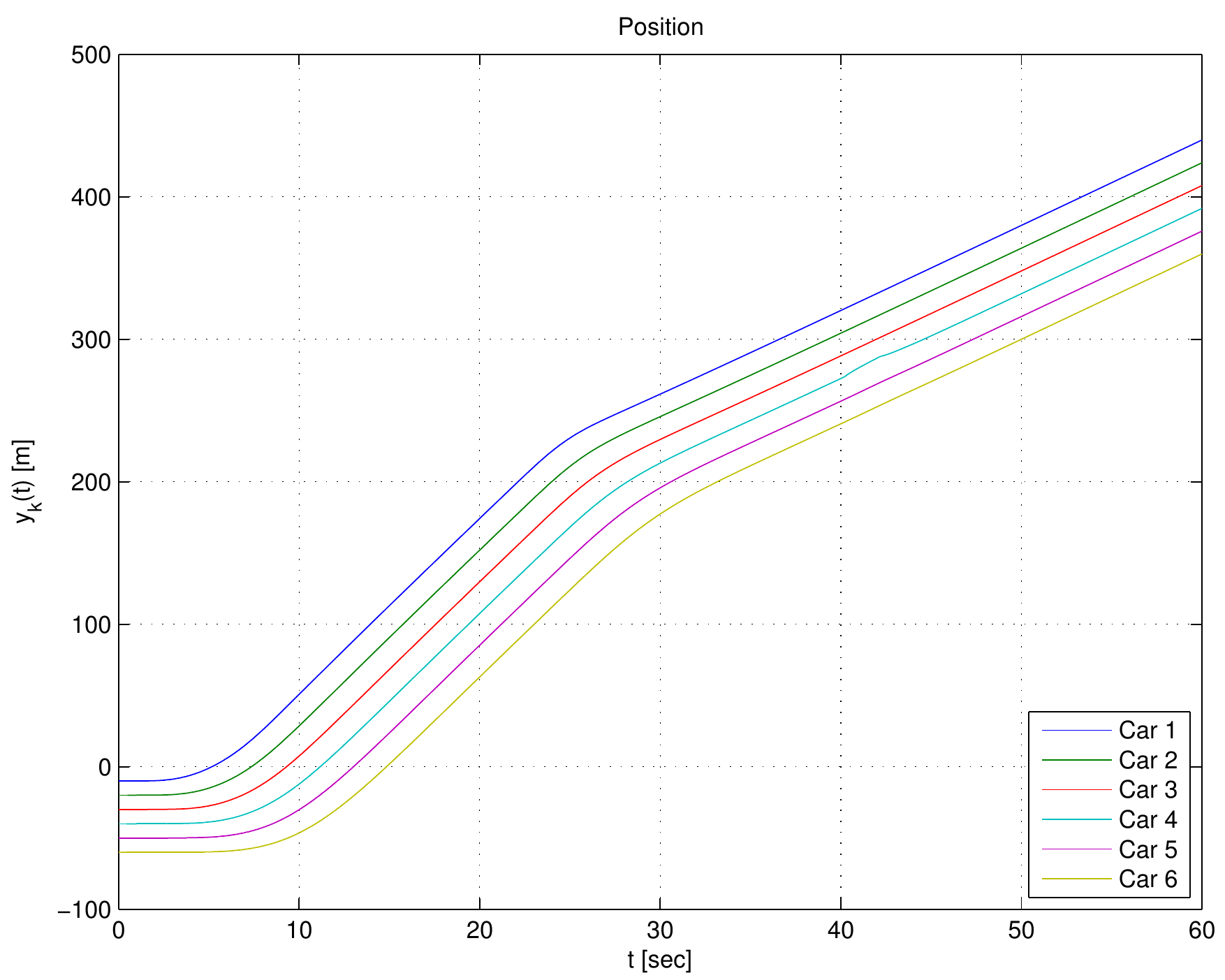}&
\includegraphics[scale=0.48]{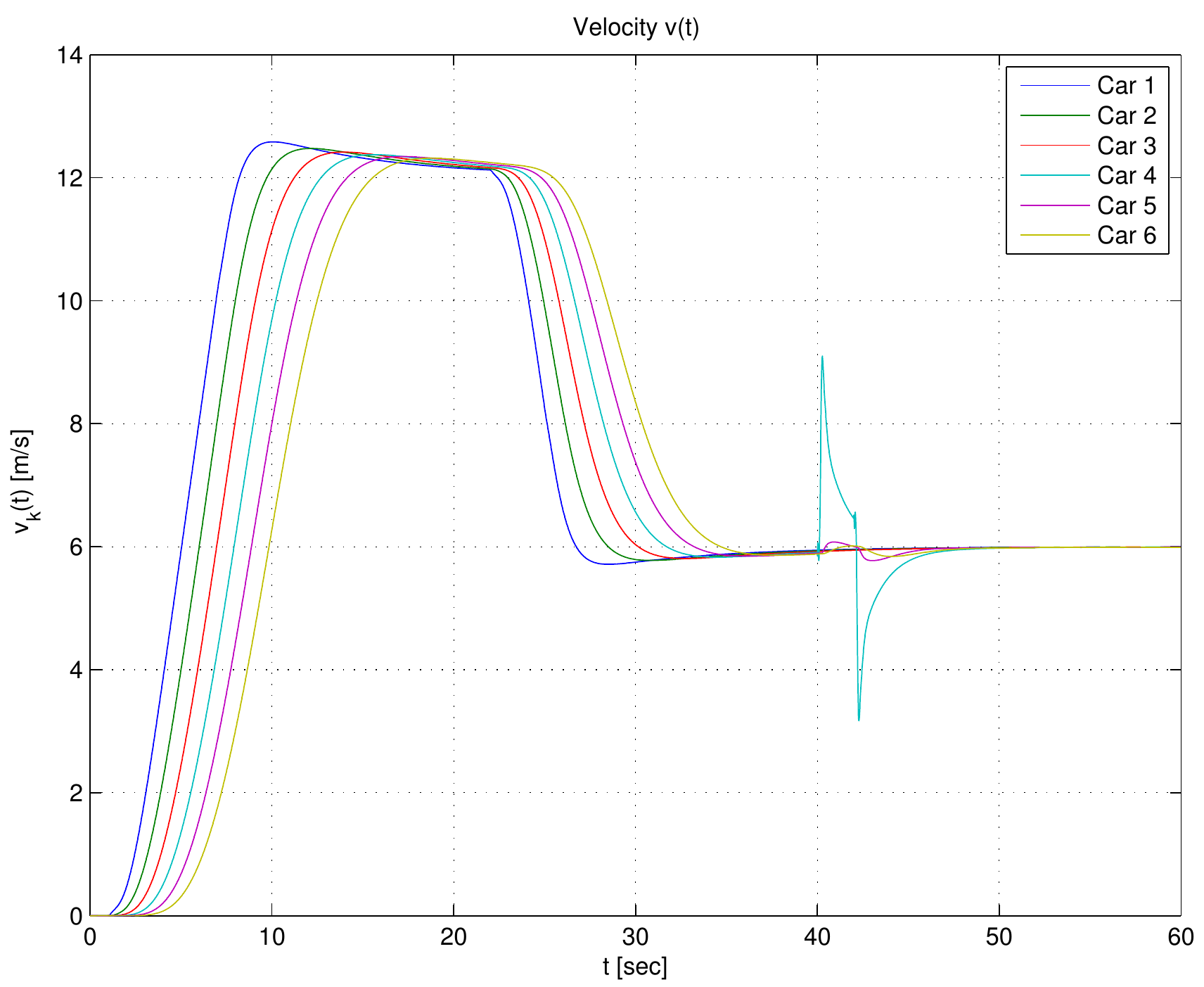}
\end{tabular}
\caption{Simulation results for Distributed Leader Information Control for Platooning (from top left): disturbances for vehicle k / acceleration at the leader; Inter-vehicle spacing errors z(t); Position y(t) (absolute value); Velocity v(t).}
\label{grafice}
\end{figure}

The leader information controllers for each vehicle were designed according to Theorem \ref{leaderinftare} and Corollary \ref{invent}, by taking  a Pade approximation for the time delays from (\ref{mai12}). The parameters $Q_{jj}$, $j=1,\dots,6$ were obtained by minimizing the practical $\mathcal{H}_\infty$ criterion from Section \ref{Praktika}, given in \eqref{synthesis}. 
The simulation results are given in Figure \ref{grafice}. The leader's control signal $u_0(t)$ (generated by the human driver in the leader vehicle) is given in the top left plot, along with a  rectangular pulse disturbance $w_4(t)$ at the $4$--th vehicle. It can be observed that the $u_0$ causes nonzero inter-vehicle spacing error  only at $z_1$, specifically the car behind the leader (vehicle with index $1$) and not at all for the vehicles $2$ and behind. The disturbance at the $4$--th vehicle affects the inter--spacing errors $z_4$ and $z_5$ only (at vehicles $4$ and $5$, respectively) and not at all for vehicle $6$.

%
%

\section{Conclusions}
\label{concl}

We have introduced  a generalization of the concept of {\em leader information controller}  for a non homogeneous platoon of  vehicles and  we have provided a Youla--like parameterization of all such stabilizing controllers. The key feature of  the leader information controller scheme is that it allows for a distributed implementation where the controller placed on each vehicle uses only locally available information. The proposed scheme is also amenable to optimal controller design using norm based costs, it guarantees string stability and it eliminates the accordion effect from the behavior of the platoon. A comprehensive analysis detailing the underlying connections with previous platooning control strategies and with existing distributed/decentralized control architectures is  performed.  We have also presented a method for exact compensation of the time delays introduced by the wireless broadcasting of information, such as to preserve all the  leader information controller performance features.

\end{document}